\documentclass[reqno, eucal]{amsart}

\usepackage[mathscr]{eucal} %  use \EuScript (\mathcal unchanged)
\usepackage[dvips]{graphicx}
\usepackage[dvips]{color}
\usepackage{amsmath,amsfonts,amssymb,amsthm,amscd}
\usepackage{epic}
\usepackage{eepic}
\usepackage{longtable}
\usepackage{array}
\usepackage{here}
\numberwithin{equation}{section}

\newtheorem{theorem}{Theorem}[section]
\newtheorem{lemma}[theorem]{Lemma}

\newtheorem{proposition}[theorem]{Proposition}
\newtheorem{corollary}[theorem]{Corollary}

\theoremstyle{definition}

\newtheorem{example}[theorem]{Example}
\newtheorem{remark}[theorem]{Remark}

\newcommand{\Z}{{\mathbb Z}}

\newcommand{\C}{{\mathbb C}}

\begin{document}

\title[Multispecies TAZRP I]{Multispecies totally asymmetric zero range process:\\
I. Multiline process and combinatorial $R$}

\author{Atsuo Kuniba}
\email{atsuo@gokutan.c.u-tokyo.ac.jp}
\address{Institute of Physics, University of Tokyo, Komaba, Tokyo 153-8902, Japan}

\author{Shouya Maruyama}
\email{maruyama@gokutan.c.u-tokyo.ac.jp}
\address{Institute of Physics, University of Tokyo, Komaba, Tokyo 153-8902, Japan}

\author{Masato Okado}
\email{okado@sci.osaka-cu.ac.jp}
\address{Department of Mathematics, Osaka City University, 
3-3-138, Sugimoto, Sumiyoshi-ku, Osaka, 558-8585, Japan}

%\date{\today}

\maketitle

\vspace{0.5cm}
\begin{center}{\bf Abstract}
\end{center}
We introduce an $n$-species totally asymmetric zero range process
($n$-TAZRP) on one-dimensional periodic lattice with $L$ sites.
It is a continuous time Markov process in which 
$n$ species of particles hop to the adjacent site 
only in one direction under the condition that 
smaller species ones have the priority to do so.
Also introduced is an $n$-line process, a companion stochastic system 
having the uniform steady state from which
the $n$-TAZRP is derived as the image by a certain projection $\pi$.
We construct the $\pi$ by a combinatorial $R$ 
of the quantum affine algebra $U_q(\widehat{sl}_L)$ and 
establish a matrix product formula of the steady state probability 
of the $n$-TAZRP
in terms of corner transfer matrices
of a $q=0$-oscillator valued vertex model. 
These results parallel the recent reformulation of 
the $n$-species totally asymmetric simple exclusion process 
($n$-TASEP) by the authors, 
demonstrating that $n$-TAZRP and $n$-TASEP are the canonical sister models
associated with the symmetric and 
the antisymmetric tensor representations of 
$U_q(\widehat{sl}_L)$ at $q=0$, respectively.

\vspace{0.5cm}

\section{Introduction}\label{sec:intro}

Zero range processes are stochastic dynamical systems 
modeling a variety of nonequilibrium phenomena in 
biology, chemistry, economics, networks, physics, sociology and so forth.
In this article and the next \cite{KMO4} 
we introduce and study a new zero range process on 
one-dimensional (1D) periodic lattice of length $L$.
There are $n$ species of particles living on the sites
with no constraint on their occupation numbers.
Particles within a site hop to the left adjacent site or  
remain unmoved under the condition that 
smaller species ones have the priority to hop.
We call it {\em $n$-species totally asymmetric zero range process} ($n$-TAZRP),
where TA refers to the unidirectional move and ZR signifies that 
the interaction of particles via the priority constraint 
works only among those occupying the same departure site.
A cheerful realization of such a system is children's play along 
a circle divided into $L$ segments.
The species of particles are interpreted as ages of the children.
They are allowed to move forward to the next segment only when
accompanying all the strictly younger fellows than themselves to look after.
As one may imagine from such an example there is a general tendency of 
{\em condensation}, whose symptom is indeed observed in our TAZRP.
See Example \ref{ex:LL}.  

There are several kinds of one-dimensional zero range processes 
studied in the literature.
They are mostly one or two-species models.
See for example \cite{EH,P,BCS} and references therein.
The $n$-TAZRP in this article and \cite{KMO4} is 
the first multispecies example which allows an explicit matrix product formula 
for the steady state probability for general $n \ge 1$.
It possesses a number of distinctive features summarized below.

(i) Our $n$-TAZRP is the image of a certain projection $\pi$ from 
another stochastic system, the {\em $n$-line process} ($n$-LP), 
which we also introduce in this paper.
It has the steady state with uniform probability distribution. 
Denoting their Markov matrices by 
$H_\mathrm{TAZRP}$ and $H_\mathrm{LP}$ respectively, 
we have the intertwining relation
\begin{align*}
\pi H_\mathrm{LP} = H_\mathrm{TAZRP}\, \pi.
\end{align*}
The map $\pi$ is a source of many intriguing features in our construction.
It is realized as a composition of a {\em combinatorial $R$} \cite{NY} of the 
quantum affine algebra $U_q(\widehat{sl}_L)$ \cite{D86,J}.
The combinatorial $R$ is a bijection between finite sets called {\em crystals} and 
arises as a quantum $R$ matrix at $q=0$ \cite{Ka1,KMN,HK}.
Systematic use of the Yang-Baxter equation \cite{Bax} 
satisfied by the combinatorial $R$ is a key maneuver in our working.
In particular the projection $\pi$ admits 
a queueing type description (Section \ref{ss:mst}) analogous 
to the Ferrari-Martin algorithm \cite{FM} for the $n$-species totally asymmetric 
simple exclusion process ($n$-TASEP).

(ii) Our main result, Theorem \ref{th:aoy}, is 
a matrix product formula of the steady state probability of the 
configuration $(\sigma_1, \ldots, \sigma_L)$ of the $n$-TAZRP:
\begin{align*}
{\mathbb P}(\sigma_1, \ldots, \sigma_L) = 
\mathrm{Tr}\bigl(X_{\sigma_1} \cdots X_{\sigma_L}\bigr).
\end{align*}
The operator $X_{\sigma}$ assigned with a local state $\sigma$ 
is a configuration sum that can be viewed as a   
{\em corner transfer matrix} \cite{Bax} of a $q=0$-oscillator valued vertex model.
It acts on the $n(n-1)/2$-fold tensor product of the $q=0$-oscillator Fock space.
See (\ref{mrn:ssi}).
The $X_{\sigma}$ can also be regarded  
as a layer-to-layer transfer matrix of a 3D lattice model, and 
${\mathbb P}(\sigma_1, \ldots, \sigma_L)$ is thereby interpreted as 
a {\em partition function} of the 3D model with prism shape under a prescribed 
boundary condition.

(iii) By extending the setting to generic $q$,  
the corner transfer matrix $X_\sigma$ is naturally embedded into a layer-to-layer
transfer matrix of a more general 3D lattice model.
Then the most local hence fundamental relation 
responsible for the steady state condition turns out to be 
the {\em tetrahedron equation} \cite{Zam80}, 
which is a 3D generalization of the Yang-Baxter equation.
The result reveals the 3D integrability in the matrix product construction.

(iv) Our $n$-TAZRP is a model in which 
the ``physical space" is of size $L$ and the ``internal space" is of size $n$.
In contrast, the internal symmetry of the combinatorial $R$ is 
$U_{q=0}(\widehat{sl}_L)$ and the system size of 
the corner transfer matrix is $n$.
See Theorem \ref{th:mnm}.
In particular the cyclic symmetry $\Z_L$ of the original lattice has been 
incorporated into the Dynkin diagram of the internal symmetry algebra 
$U_q(\widehat{sl}_L)$.
In this sense our approach captures the {\em cross channel} of the original problem
where the two kinds of spaces and symmetries are interchanged.
It is a manifestation of the {\em rank-size duality} commonly recognized 
for a class of 3D systems associated with the tetrahedron equation \cite{BS, KOS}.
We stress that such a hidden 3D structure can be elucidated only
by a systematic investigation on the multispecies case $n\ge 1$. 
An alternative approach via the direct channel 
based on some rank $n$ internal symmetry algebra
is yet to be undertaken.

(v) The whole story about the $n$-TAZRP in this paper and \cite{KMO4}
is closely parallel with the recent result 
on the $n$-TASEP by the authors \cite{KMO, KMO2}.
In fact the $n$-TAZRP and the $n$-TASEP turn out to be 
the canonical sister models associated with the 
symmetric and the antisymmetric tensor representations of 
$U_q(\widehat{sl}_L)$ at $q=0$, respectively.
The combinatorial $R$'s for both of them had been obtained in \cite{NY}.
In terms of the 3D picture, the two models are associated with 
3D $R$-operator and the 3D $L$-operator, respectively.
They are distinguished solutions to the tetrahedron equation which  
have a rich background going back to the 
representation theory of the quantized algebra of functions \cite{KV}.
See \cite{BS, KOS, Ku} and references therein.

In this paper we will demonstrate the features 
(i) and (ii) of combinatorial nature mainly, and leave the 
issue (iii) related to the tetrahedron equation to the subsequent paper \cite{KMO4}. 
Although the main idea comes from the crystal base theory,
a theory of quantum groups at $q=0$ \cite{Ka1, HK},
the article has been designed to be readable without knowledge of it.

In Section \ref{sec:taz}
the $n$-TAZRP is defined and examples of the steady states are presented.
In Section \ref{sec:mp}  the $n$-line process is introduced 
on the set $B({\bf m})$ 
which is the crystal of the $n$-fold tensor product of  
the symmetric tensor representation of $U_q(\widehat{sl}_L)$. 
It is shown that the steady state of the 
$n$-line process has the uniform probability distribution 
(Theorem \ref{th:akn}).
In Section \ref{sec:pr}
the projection $\pi$ from the $n$-line process to the 
$n$-TAZRP is constructed from a composition of the combinatorial $R$. 
It is also described in terms of a multiple queueing 
type algorithm in Section \ref{ss:mst}, which contrasts with 
the analogous procedure for TASEP \cite{FM,KMO}. 
The steady state of the $n$-TAZRP is the image of the uniform state 
in the $n$-line process.
In Section \ref{sec:f} 
a matrix product formula for the steady state probability 
is derived for the $n$-TAZRP, which is expressed 
by corner transfer matrices of $q=0$-oscillator valued 
vertex model.

Throughout the paper we use the characteristic function $\theta$ defined by 
$\theta(\text{true}) = 1, \theta(\text{false}) = 0$.

\section{$n$-TAZRP}\label{sec:taz}

\subsection{Definition of $n$-TAZRP}
Consider a periodic one-dimensional chain $\Z_L$ with $L$ sites.
Each site $i \in \Z_L$ is assigned with a local state 
$\sigma_i=(\sigma_i^1,\ldots, \sigma^n_i)\in (\Z_{\ge 0})^n$ which is 
interpreted as an assembly of $n$ species of particles as
\begin{equation}\label{cie}
\begin{picture}(260,25)(-80,0)

\put(-1,5){$\overbrace{1 \ldots 1}^{\sigma^1_i} \,
\overbrace{2 \ldots 2}^{\sigma^2_i} \,
\ldots \,
\overbrace{n \ldots n}^{\sigma^n_i} $}
\put(7,0){\put(-16,1){\line(1,0){115}}\put(-16,1){\line(0,1){12}}
\put(99,1){\line(0,1){12}}}

\end{picture}
\end{equation}
The ordering of particles within a site does not matter.
A local state $\alpha$ is specified uniquely either by 
{\em multiplicity representation}
$\alpha=(\alpha^1,\ldots, \alpha^n) \in (\Z_{\ge 0})^n$ as above or 
{\em multiset representation}
$\alpha=(\alpha_1,\ldots, \alpha_r) \in [1,n]^r$ 
with $1 \le \alpha_1 \le \cdots \le \alpha_r\le n$.
They are related by
$\alpha^a = \#\{j \in [1,r]\mid \alpha_j = a\}$ and 
$r = |\alpha|:=\alpha^1 + \cdots + \alpha^n$.

Let $({\alpha}, {\beta})$ 
and $({\gamma},{\delta})$ be pairs of local states.
Let $(\beta_1, \ldots, \beta_r)$ 
be the multiset representation of $\beta$, hence 
$1 \le \beta_1 \le \cdots \le \beta_r \le n$.
For the two pairs we define $>$ by
\begin{align}\label{kyk}
({\alpha}, {\beta}) > ({\gamma},{\delta})
\overset{\text{def}}{\Longleftrightarrow}
\gamma = \alpha \cup \{\beta_1,\ldots, \beta_k\},\;
\delta=(\beta_{k+1},\ldots, \beta_r)\;\;
\text{for some}\; k \in [1,r],
\end{align}
where $\alpha \cup \{\beta_1,\ldots, \beta_k\}$ is a union as a multiset.
For instance in multiset representation we have\footnote{
Here and in what follows, 
a multiset (set accounting for multiplicity of elements), say  $\{1,1,3,5,6\}$,
is abbreviated to $11356$, which does not cause a confusion since
all the examples in this paper shall be concerned with the case $n\le 9$.}
\begin{equation}\label{ykw}
\begin{split}
(1356,114)&> (11356,14),  (111356,4), (1113456,\emptyset),\\
(235,12446) &> (1235,2446), (12235,446), (122345,46), (1223445,6), (12234456, \emptyset),\\
(\emptyset,225) & > (2,25), (22,5), (225,\emptyset),\\
(344,\emptyset) &> \text{none}.
\end{split}
\end{equation} 
We let $({\alpha}, {\beta}) \ge ({\gamma},{\delta})$ 
mean $({\alpha}, {\beta}) > ({\gamma},{\delta})$ 
or $({\alpha}, {\beta}) = ({\gamma},{\delta})$.

By $n$-TAZRP
we mean a stochastic process on $\Z_L$ 
in which neighboring pairs of local states
$(\sigma_i, \sigma_{i+1})=(\alpha, \beta)$ 
change into $(\gamma, \delta)$ such that 
$({\alpha}, {\beta}) > ({\gamma},{\delta})$
with a uniform transition rate.
For example the first line in (\ref{ykw}) implies that 
the following local transitions take place with an equal rate:
\begin{align*}
\begin{picture}(500,105)(-130,-77)
\put(10,3.5){1356} \put(53,3.5){114}
\put(97,0){\put(10,3.5){11356} \put(58,3.5){14}}
\put(38,25){1}
\put(20,22){\vector(0,-1){8}}\put(20,22){\line(1,0){40}}
\put(60,22){\line(0,-1){7}}
\put(0,0){\line(1,0){80}}
\put(0,0){\line(0,1){10}}
\put(40,0){\line(0,1){10}}
\put(80,0){\line(0,1){10}}
\put(85, 6){\vector(1,0){10}}
\put(100,0){
\put(0,0){\line(1,0){80}}
\put(0,0){\line(0,1){10}}
\put(40,0){\line(0,1){10}}
\put(80,0){\line(0,1){10}}}
%
%2
\put(0,-40){
\put(10,3.5){1356} \put(53,3.5){114}
\put(97,0){\put(8,3.5){111356} \put(58,3.5){4}}
\put(36,25){11}
\put(20,22){\vector(0,-1){8}}\put(20,22){\line(1,0){40}}
\put(60,22){\line(0,-1){7}}
\put(0,0){\line(1,0){80}}
\put(0,0){\line(0,1){10}}
\put(40,0){\line(0,1){10}}
\put(80,0){\line(0,1){10}}
\put(85, 6){\vector(1,0){10}}
\put(100,0){
\put(0,0){\line(1,0){80}}
\put(0,0){\line(0,1){10}}
\put(40,0){\line(0,1){10}}
\put(80,0){\line(0,1){10}}}
}
%3
\put(0,-80){
\put(10,3.5){1356} \put(53,3.5){114}
\put(97,0){\put(6,3.5){1113456} }
\put(34,25){114}
\put(20,22){\vector(0,-1){8}}\put(20,22){\line(1,0){40}}
\put(60,22){\line(0,-1){7}}
\put(0,0){\line(1,0){80}}
\put(0,0){\line(0,1){10}}
\put(40,0){\line(0,1){10}}
\put(80,0){\line(0,1){10}}
\put(85, 6){\vector(1,0){10}}
\put(100,0){
\put(0,0){\line(1,0){80}}
\put(0,0){\line(0,1){10}}
\put(40,0){\line(0,1){10}}
\put(80,0){\line(0,1){10}}}
}
\end{picture}
\end{align*}
In general we let $\tau^k_i: (\sigma_1,\ldots, \sigma_L) \mapsto 
(\sigma'_1,\ldots, \sigma'_L)$ denote the transition (\ref{kyk}) in which smaller species 
$k$ particles move from the 
$(i+1)$-th site to the $i$-th site with no change elsewhere:
\begin{equation}\label{kgc2}
\begin{picture}(120,50)(7,-13)

\put(-140,5){$\tau^k_i :$}

%%LHS

\put(-85,5){${\alpha}_1\, \ldots \,{\alpha}_s$}
\put(-13,5){${\beta}_1 \ldots \ldots .\,. \, {\beta}_r$}
\put(-110,0){

\put(65,26){${\beta}_1 \ldots {\beta}_k$}
\put(41.5,22){\vector(0,-1){9}}\put(41.5,22){\line(1,0){83}}
\put(124.5,22){\line(0,-1){8}}

\put(0,0){\line(1,0){166}}\put(0,0){\line(0,1){12}}
\put(83,0){\line(0,1){12}}\put(166,0){\line(0,1){12}}

\put(38,-12){$\sigma_i$}\put(115,-11){$\sigma_{i+1}$}

}

\put(63,6){\vector(1,0){15}}
%%RHS

\put(-10,0){
\put(100,5){${\alpha}_1 \ldots {\alpha}_s \,
{\beta}_1 \ldots {\beta}_k$}
\put(195,5){${\beta}_{k+1} \,\ldots \,{\beta}_r$}
\put(95,0){\put(0,0){\line(1,0){166}}\put(0,0){\line(0,1){12}}
\put(83,0){\line(0,1){12}}\put(166,0){\line(0,1){12}}
\put(38,-12){$\sigma'_i$}\put(115,-11){$\sigma'_{i+1}$}
}}
\end{picture}
\end{equation}
where $1 \le {\beta}_1 \le \cdots \le 
{\beta}_r \le n, \, k \in [1,r], i \in \Z_L$, and 
$\sigma'_j = \sigma_j$ for $j\neq i, i+1$.
For a later convenience 
we extend $\tau^k_i$ to all $k \in \Z_{\ge 1}$
by setting $\tau^k_i = 1$ for $k>r$ which means to move no particle.

This dynamics is {\em totally asymmetric} in that 
particles can hop only to the left adjacent site.
Their interaction is of {\em zero range} in that 
the hopping priority for smaller species particles is respected 
only among those occupying the same site.
There is no constraint on the status of the destination site 
nor number of particles that hop at a transition.  
A pair $(\alpha, \beta)$ of adjacent local states has $|\beta|$ 
possibilities to change into.

The $n$-TAZRP dynamics obviously preserves 
the number of particles of each species.
Thus we introduce {\em sectors} labeled with 
{\em multiplicity} ${\bf m}=(m_1,\ldots, m_n) \in (\Z_{\ge 0})^n$ of 
the species of particles:
\begin{align}\label{kgc}
S({\bf m}) = 
\{{\boldsymbol \sigma}=(\sigma_1,\ldots, \sigma_L)\mid
\sigma_i = (\sigma^{1}_i, \ldots, \sigma^{n}_i)  \in (\Z_{\ge 0})^n,\;
\sum_{i=1}^L \sigma^{a}_i=m_a,\forall a \in [1,n]\}.
\end{align}
A configuration will also be written as 
${\boldsymbol \sigma}=(\sigma^{a}_i)$.
A sector $S({\bf m})$ such that $m_a \ge 1$ for all $a \in [1,n]$
is called {\em basic}.
Non-basic sectors are equivalent to a basic sector for $n'$-TAZRP with some 
$n'<n$ by a suitable relabeling of species.
Thus we shall exclusively deal with basic sectors in this paper.

A local state $\sigma_i$ in (\ref{kgc}) 
can take $N=\prod_{a=1}^n(m_a+1)$ possibilities in view of (\ref{cie}).
Let $\{|{\boldsymbol \sigma} \rangle = |\sigma_1,\ldots, \sigma_L\rangle\}$ 
be a basis of
$(\C^N)^{\otimes L}$.
Denoting by ${\mathbb P}(\sigma_1,\ldots, \sigma_L; t)$ the probability of finding 
the system in the configuration 
${\boldsymbol \sigma}=(\sigma_1,\ldots, \sigma_L)$ at time $t$,  we set 
\begin{align*}
|P(t)\rangle
= \sum_{{\boldsymbol \sigma} \in S({\bf m})}
{\mathbb P}(\sigma_1,\ldots, \sigma_L; t)|\sigma_1,\ldots, \sigma_L\rangle.
\end{align*}
This actually belongs to a
subspace of $(\C^N)^{\otimes L}$ of dimension 
$\# S({\bf m}) = \prod_{a=1}^n\binom{L+m_a-1}{m_a}$ which is in general much 
smaller than $N^L$ reflecting the constraint in (\ref{kgc}).

Our $n$-TAZRP is a stochastic system 
governed by the continuous-time master equation
\begin{align*}
\frac{d}{dt}|P(t)\rangle
= H_{\mathrm{TAZRP}} |P(t)\rangle,
\end{align*}
where the Markov matrix has the form
\begin{align}\label{hrk2}
H_{\mathrm{TAZRP}}  = \sum_{i \in \Z_L} h_{i,i+1},\qquad
h |\alpha, \beta\rangle =\sum_{\gamma,\delta}h^{\gamma,\delta}_{\alpha,\beta}
|\gamma,\delta\rangle.
\end{align}
Here $h_{i,i+1}$ is the local Markov matrix that  
acts as $h$ on the $i$-th and the $(i+1)$-th components nontrivially and 
as the identity elsewhere.
If the transition rate of the adjacent pair of local states
$(\alpha, \beta) \rightarrow (\gamma,\delta)$ is denoted by 
$w(\alpha\beta \rightarrow \gamma\delta)$,
the matrix element of the Markov matrix is given by 
$h^{\gamma,\delta}_{\alpha,\beta}
= w(\alpha\beta \rightarrow \gamma\delta)
-\theta\bigl((\alpha,\beta)=(\gamma,\delta)\bigr)
\sum_{\gamma',\delta'}w(\alpha\beta \rightarrow \gamma'\delta')$.
Our $n$-TAZRP corresponds to the choice $w(\alpha\beta \rightarrow \gamma\delta) = 
\theta\bigl((\alpha,\beta)>(\gamma,\delta)\bigr)$, therefore the general formula,
which is independent of $w(\alpha\beta \rightarrow \alpha\beta)$, gives
\begin{align*}
h^{\gamma,\delta}_{\alpha,\beta}=
\begin{cases}
1 & \text{if }\;(\alpha, \beta) > (\gamma,\delta),\\
-|\beta| & \text{if } \;(\alpha, \beta)  = (\gamma,\delta),\\
0 & \text{otherwise}.
\end{cases}
\end{align*}
The Markov matrix (\ref{hrk2}) is expressed as 
\begin{equation}\label{ymi:s}
H_{\mathrm{TAZRP}}  = \sum_{i \in \Z_L}\sum_{k\ge 1}(\tau^k_i-1),
\end{equation}
which is actually a finite sum due to the convention 
explained after (\ref{kgc2}).

\subsection{Steady state}
As time goes on, the distribution of the particles converges to the state
that we consider from now on.
Given a system size $L$ and a sector $S({\bf m})$
there is a unique vector 
\begin{align}\label{ymi:cu}
|\bar{P}_L({\bf m})\rangle = \sum_{{\boldsymbol \sigma} \in S({\bf m})}
\mathbb{P}({\boldsymbol \sigma} ) |{\boldsymbol \sigma} \rangle
\end{align} 
up to a normalization, called the {\em steady state}, 
which satisfies $H_{\mathrm{TAZRP}}  |\bar{P}_L({\bf m})\rangle=0$ hence is 
time-independent.
In what follows we will always take ${\mathbb P}({\boldsymbol \sigma})$ 
so that 
\begin{align*}
\sum_{{\boldsymbol \sigma} \in S({\bf m})}
{\mathbb P}({\boldsymbol \sigma}) = \# B({\bf m})
\end{align*}
holds, where $\# B({\bf m})$ is given by (\ref{mkr:arkr}) and (\ref{mkr:akci}).
The unnormalized 
${\mathbb P}({\boldsymbol \sigma})$ will be called the steady state probability
by abusing the terminology. 
The properly normalized one is equal to 
${\mathbb P}_\text{normalized}({\boldsymbol \sigma})=
{\mathbb P}({\boldsymbol \sigma})/(\# B({\bf m}))$.
This convention is convenient for our working below in that 
${\mathbb P}({\boldsymbol \sigma}) \in \Z_{\ge 1}$ holds
as we will see in Theorem \ref{th:szk} and (\ref{szk:i}).

The steady state for $1$-TAZRP is trivial 
under the present periodic boundary condition in that  
all the configurations are realized with an equal probability.

\begin{example}\label{ex:LL}
We present the steady state in small sectors of
2-TAZRP and 3-TAZRP in the form
\begin{align*}
|\bar{P}_L({\bf m})\rangle = |\xi_L({\bf m})\rangle
+ C|\xi_L({\bf m})\rangle + \cdots + 
C^{L-1} |\xi_L({\bf m})\rangle
\end{align*}
respecting the symmetry 
$H_{\mathrm{TAZRP}} C=CH_{\mathrm{TAZRP}} $ under the $\Z_L$ cyclic shift
 $C: |\sigma_1, \sigma_2,\ldots, \sigma_L\rangle \mapsto 
 |\sigma_L, \sigma_1, \ldots, \sigma_{L-1}\rangle$.
The choice of the vector $|\xi_L({\bf m})\rangle$ is not unique.
We employ multiset representation like 
$|\emptyset, 3, 122\rangle$, which would have looked as 
$|000,001,120\rangle$ in the multiplicity representation for 3-TAZRP.

For the 2-TAZRP  one has
\begin{align*}
|\xi_2(1,1)\rangle&= 2|\emptyset, 12\rangle + | 1, 2\rangle,\\
|\xi_3(1,1)\rangle& = 3|\emptyset, \emptyset, 12\rangle +
2|\emptyset, 1, 2\rangle +
|\emptyset, 2, 1\rangle,\\
|\xi_4(1,1)\rangle& = 4|\emptyset, \emptyset, \emptyset, 12\rangle +
3|\emptyset, \emptyset, 1,2 \rangle +
2|\emptyset, 1,\emptyset, 2 \rangle +
|\emptyset, \emptyset, 2,1 \rangle,\\
|\xi_2(2,1)\rangle& =
2|\emptyset, 112\rangle + |1,12\rangle + |2, 11\rangle,\\
|\xi_3(2,1)\rangle &=
3|\emptyset, \emptyset, 112\rangle +
2|\emptyset, 1,12\rangle +
|\emptyset, 2, 11\rangle +
2|\emptyset, 11, 2\rangle +
|\emptyset, 12,1\rangle +
|1,1,2\rangle,\\
|\xi_4(2,1)\rangle &=
2|\emptyset, 1, 1, 2\rangle 
+ |\emptyset, 1, 2, 1\rangle + 
 2 |\emptyset, 1, \emptyset, 12\rangle + |\emptyset, 2, 1, 1\rangle + 
 2 |\emptyset, 2, \emptyset, 11\rangle + 
 3 |\emptyset, \emptyset, 1, 12\rangle \\
&+ 
 |\emptyset, \emptyset, 2, 11\rangle + 
 3 |\emptyset, \emptyset, 11, 2\rangle + 
 |\emptyset, \emptyset, 12, 1\rangle + 
 4 |\emptyset, \emptyset, \emptyset, 112\rangle,\\
|\xi_2(1,2)\rangle &=
3|\emptyset,122\rangle + |1,22\rangle + 2 |2,12\rangle,\\
|\xi_3(1,2)\rangle &=
6|\emptyset, \emptyset, 122\rangle +
3|\emptyset, 1, 22\rangle +
3|\emptyset, 2, 12\rangle +
5|\emptyset, 12, 2\rangle +
|\emptyset, 22, 1\rangle +
2|1,2,2\rangle,\\
|\xi_4(1,2)\rangle &=
5 |\emptyset, 1, 2, 2\rangle + 3 |\emptyset, 1, \emptyset, 22\rangle + 
 3 |\emptyset, 2, 1, 2\rangle + 2 |\emptyset, 2, 2, 1\rangle + 
 7 |\emptyset, 2, \emptyset, 12\rangle + 
 6 |\emptyset, \emptyset, 1, 22\rangle \\
&+ 
 4 |\emptyset, \emptyset, 2, 12\rangle + 
 9 |\emptyset, \emptyset, 12, 2\rangle + 
 |\emptyset, \emptyset, 22, 1\rangle + 
 10 |\emptyset, \emptyset, \emptyset, 122\rangle,\\
|\xi_2(3,1)\rangle&=
2 |\emptyset, 1112\rangle + |{1}, 112\rangle + |{2}, 111\rangle + 
 |11, 12\rangle,\\
|\xi_3(3,1)\rangle&=
3 |\emptyset, \emptyset, 1112\rangle + 2 |\emptyset, {1}, 112\rangle + 
 |\emptyset, {2}, 111\rangle + 2 |\emptyset, 11, 12\rangle + 
 |\emptyset, 12, 11\rangle \\
&+ 2 |\emptyset, 111, {2}\rangle + 
 |\emptyset, 112, {1}\rangle + |{1}, {1}, 12\rangle + 
 |{1}, {2}, 11\rangle + |{1}, 11, {2}\rangle,\\
|\xi_2(2,2)\rangle&=
3|\emptyset,1122\rangle + |1 , 122\rangle + 2|2,112\rangle + |11,22\rangle
+\textstyle\frac{1}{2}|12,12\rangle, \\
|\xi_3(2,2)\rangle&=
|1, 1, 22\rangle + |1, 2, 12\rangle + 2 |1, 12, 2\rangle + 2 |2, 2, 11\rangle + 
 3 |\emptyset, 1, 122\rangle + 3 |\emptyset, 2, 112\rangle \\
 &+ 
 3 |\emptyset, 11, 22\rangle + 2 |\emptyset, 12, 12\rangle + 
 |\emptyset, 22, 11\rangle + 5 |\emptyset, 112, 2\rangle + 
 |\emptyset, 122, 1\rangle + 6 |\emptyset, \emptyset, 1122\rangle,\\
|\xi_2(1,3)\rangle&=
|1, 222\rangle + 3 |2, 122\rangle + 2 |12, 22\rangle + 4 |\emptyset, 1222\rangle,\\
|\xi_3(1,3)\rangle&=
2 |1, 2, 22\rangle + 3 |1, 22, 2\rangle + 5 |2, 2, 12\rangle + 
 4 |\emptyset, 1, 222\rangle + 6 |\emptyset, 2, 122\rangle + 
 7 |\emptyset, 12, 22\rangle \\
 &+ 3 |\emptyset, 22, 12\rangle + 
 9 |\emptyset, 122, 2\rangle + |\emptyset, 222, 1\rangle + 
 10 |\emptyset, \emptyset, 1222\rangle.
\end{align*}
Although $|\xi_2(2,2)\rangle$ contains 
$\textstyle\frac{1}{2}|12,12\rangle$,
the coefficients in
$|\bar{P}_2(2,2)\rangle$ are integers  
as this configuration is an order 2 fixed point of $C$. 

For the 3-TAZRP one has
\begin{align*}
|\xi_2(1,1,1)\rangle&=
2 |1, 23\rangle + |2, 13\rangle + 3 |3, 12\rangle + 6 |\emptyset, 123\rangle
,\\
|\xi_3(1,1,1)\rangle&=
5 |1, 2, 3\rangle + |1, 3, 2\rangle + 9 |\emptyset, 1, 23\rangle + 
 3 |\emptyset, 2, 13\rangle + 6 |\emptyset, 3, 12\rangle + 
 12 |\emptyset, 12, 3\rangle \\
 &+ 3 |\emptyset, 13, 2\rangle + 
 3 |\emptyset, 23, 1\rangle + 18 |\emptyset, \emptyset, 123\rangle,\\
|\xi_4(1,1,1)\rangle&=
17 |\emptyset, 1, 2, 3\rangle + 3 |\emptyset, 1, 3, 2\rangle + 
 12 |\emptyset, 1, \emptyset, 23\rangle + 3 |\emptyset, 2, 1, 3\rangle + 
 7 |\emptyset, 2, 3, 1\rangle + 8 |\emptyset, 2, \emptyset, 13\rangle \\
&+ 
 9 |\emptyset, 3, 1, 2\rangle + |\emptyset, 3, 2, 1\rangle + 
 20 |\emptyset, 3, \emptyset, 12\rangle + 
 24 |\emptyset, \emptyset, 1, 23\rangle + 
 6 |\emptyset, \emptyset, 2, 13\rangle + 
 10 |\emptyset, \emptyset, 3, 12\rangle \\
&+ 
 30 |\emptyset, \emptyset, 12, 3\rangle + 
 6 |\emptyset, \emptyset, 13, 2\rangle + 
 4 |\emptyset, \emptyset, 23, 1\rangle + 
 40 |\emptyset, \emptyset, \emptyset, 123\rangle,\\
|\xi_2(2,1,1)\rangle&=
2 |1, 123\rangle + |2, 113\rangle + 3 |3, 112\rangle + 2 |11, 23\rangle + 
 |12, 13\rangle + 6 |\emptyset, 1123\rangle
,\\
|\xi_3(2,1,1)\rangle&=
3 |1, 1, 23\rangle + 2 |1, 2, 13\rangle + |1, 3, 12\rangle + 5 |1, 12, 3\rangle + 
 |1, 13, 2\rangle + 5 |2, 3, 11\rangle + |2, 11, 3\rangle \\
 &+ 
 9 |\emptyset, 1, 123\rangle + 3 |\emptyset, 2, 113\rangle + 
 6 |\emptyset, 3, 112\rangle + 9 |\emptyset, 11, 23\rangle + 
 3 |\emptyset, 12, 13\rangle + 3 |\emptyset, 13, 12\rangle \\
 &+ 
 3 |\emptyset, 23, 11\rangle + 12 |\emptyset, 112, 3\rangle + 
 3 |\emptyset, 113, 2\rangle + 3 |\emptyset, 123, 1\rangle + 
 18 |\emptyset, \emptyset, 1123\rangle
,\\
|\xi_2(1,2,1)\rangle&=
2 |1, 223\rangle + 2 |2, 123\rangle + 4 |3, 122\rangle + 3 |12, 23\rangle + 
 |13, 22\rangle + 8 |\emptyset, 1223\rangle
,\\
|\xi_3(1,2,1)\rangle&=
5 |1, 2, 23\rangle + |1, 3, 22\rangle + 7 |1, 22, 3\rangle + 2 |1, 23, 2\rangle + 
 3 |2, 2, 13\rangle + 9 |2, 3, 12\rangle + 3 |2, 12, 3\rangle \\
 &+ 
 12 |\emptyset, 1, 223\rangle + 8 |\emptyset, 2, 123\rangle + 
 10 |\emptyset, 3, 122\rangle + 17 |\emptyset, 12, 23\rangle + 
 4 |\emptyset, 13, 22\rangle + 3 |\emptyset, 22, 13\rangle \\
 &+ 
 6 |\emptyset, 23, 12\rangle + 20 |\emptyset, 122, 3\rangle + 
 7 |\emptyset, 123, 2\rangle + 3 |\emptyset, 223, 1\rangle + 
 30 |\emptyset, \emptyset, 1223\rangle
,\\
|\xi_2(1,1,2)\rangle&=
3 |1, 233\rangle + |2, 133\rangle + 8 |3, 123\rangle + 4 |12, 33\rangle + 
 2 |13, 23\rangle + 12 |\emptyset, 1233\rangle
,\\
|\xi_3(1,1,2)\rangle&=
9 |1, 2, 33\rangle + 3 |1, 3, 23\rangle + 17 |1, 23, 3\rangle + 
 |1, 33, 2\rangle + 7 |2, 3, 13\rangle + 3 |2, 13, 3\rangle + 
 20 |3, 3, 12\rangle \\
 &+ 24 |\emptyset, 1, 233\rangle + 
 6 |\emptyset, 2, 133\rangle + 30 |\emptyset, 3, 123\rangle + 
 30 |\emptyset, 12, 33\rangle + 12 |\emptyset, 13, 23\rangle + 
 8 |\emptyset, 23, 13\rangle \\
 &+ 10 |\emptyset, 33, 12\rangle + 
 50 |\emptyset, 123, 3\rangle + 4 |\emptyset, 133, 2\rangle + 
 6 |\emptyset, 233, 1\rangle + 60 |\emptyset, \emptyset, 1233\rangle.
\end{align*}
As these coefficients indicate, steady states are highly nontrivial for 
$n$-TAZRP with $n\ge 2$.
The main issue in the subsequent sections is to characterize them
in terms of the $n$-line process and the combinatorial $R$.
Note that the maximally localized configuration like 
$60 |\emptyset, \emptyset, 1233\rangle$  
just above and their cyclic permutations have the largest probability.
It is a symptom of {\em condensation} generally expected 
in the zero-range processes \cite{EH}.
See (\ref{kyk:sn}) for the result on the general case. 
\end{example}

\section{$n$-line process}\label{sec:mp}

\subsection{$n$-line states}\label{ss:nms}

Fix the multiplicity array ${\bf m}$ and 
$\ell_1, \ldots, \ell_n$ as
\begin{align}\label{mkr:akci}
{\bf m}=(m_1,\ldots, m_n) \in (\Z_{\ge 1})^n,\qquad
\ell_a = m_a+m_{a+1}+ \cdots + m_n\quad (1 \le a \le n)
\end{align}
so that $\ell_1> \ell_2 > \cdots > \ell_n\ge 1$. 
Associated with these data we introduce the finite sets
\begin{equation}\label{mkr}
\begin{split}
B({\bf m}) &= B_{\ell_1} \otimes \cdots \otimes B_{\ell_n},\\
B_\ell &=\{(x_1,\ldots, x_L)\in (\Z_{\ge 0})^L\mid
x_1+ \cdots + x_L = \ell\}\quad (\ell \in \Z_{\ge 1}),
\end{split}
\end{equation} 
where $\otimes$ may just be regarded as the product of sets.
By the definition we have 
\begin{align}\label{mkr:arkr}
\# B({\bf m}) = \prod_{a=1}^n \binom{L-1+\ell_a}{\ell_a}.
\end{align}
The sets $B_\ell$ and $B({\bf m})$ will be called {\em crystals}
bearing in mind, though not utilized significantly in this paper, 
that they can be endowed with the structure of the 
$U_q(\widehat{sl}_L)$-{\em crystals} 
of the $\ell$-fold symmetric tensor representation 
and their tensor product, respectively \cite{Ka1,HK}.

Our $n$-line process is a stochastic dynamical system on $B({\bf m})$.
Its elements will be called {\em $n$-line states} and 
denoted by
\begin{align}\label{mkr:gk}
{\bf x} = {\bf x}^1 \otimes \cdots \otimes {\bf x}^n \in B({\bf m}),\quad
{\bf x}^a = (x^a_1,\ldots, x^a_L) \in B_{\ell_a}.
\end{align}
We will represent it as an array ${\bf x}=({\bf x}^a)$ or as the $n \times L$ tableau:
\begin{align*}
{\bf x}=(x^a_i) = 
\begin{pmatrix}
x^n_1 & \cdots & x^n_L\\
\vdots & \ddots & \vdots\\
x^1_1 & \cdots & x^1_L
\end{pmatrix},
\end{align*}
which is conveniently depicted by a dot diagram.  
For instance when $(n,L)=(4,3)$, it looks as
\begin{equation}\label{hrk}
\begin{picture}(200,75)(-40,-6)
\put(-100,25){
${\bf x}=\begin{pmatrix}
2 & 0 & 1\\
1 & 2 & 2\\
1 & 1 & 4\\
3 & 2 & 3
\end{pmatrix}= $}
\multiput(0,0)(0,15){5}{\put(0,0){\line(1,0){45}}}
\multiput(0,0)(15,0){4}{\put(0,0){\line(0,1){60}}}

\put(2,50){$\bullet$}\put(8,50){$\bullet$}
\put(35,50){$\bullet$}

\put(5,35){$\bullet$}
\put(17,35){$\bullet$}\put(23,35){$\bullet$}
\put(32,35){$\bullet$}\put(38,35){$\bullet$}

\put(5,20){$\bullet$}
\put(20,20){$\bullet$}
\put(32,23){$\bullet$}\put(38,23){$\bullet$}
\put(32,17){$\bullet$}\put(38,17){$\bullet$}

\put(2,2){$\bullet$}\put(5,8){$\bullet$}\put(9,2){$\bullet$}
\put(17,5){$\bullet$}\put(23,5){$\bullet$}
\put(32,2){$\bullet$}\put(35,8){$\bullet$}\put(39,2){$\bullet$}

\put(50,25) {
$\in B_8 \otimes B_6 \otimes B_5 \otimes B_3 = B(2,1,2,3)$.}

\end{picture}
\end{equation}
The dynamics of our $n$-line process will be 
described as a motion of dots.
There is another multiline process relevant to the $n$-TASEP \cite{FM} 
which was reformulated in terms of crystals in \cite{KMO}.
The basic set there is 
$B^\ell = \{(x_1,\ldots, x_L) \in \{0,1\}^L\mid x_1+\cdots + x_L = \ell\}$
corresponding to the antisymmetric tensor representation
instead of the $B_\ell$ in (\ref{mkr}).

\subsection{Auxiliary sets ${\mathcal A}$ and ${\mathcal B}$}
\label{ss:ab}
 
Let $[a,b]$ denote the set $\{a,a+1,\ldots, b\}$ for the integers 
$a\le b$.
We set $(x)_+ = \max(x,0)$ so that $(x)_+- (-x)_+ = x$ holds.

Given an $n$-line state ${\bf x} \in B({\bf m})$
we attach to it the sets 
${\mathcal A}_{\bf x}$ and ${\mathcal B}_{\bf x}$ defined by
\begin{align}
{\mathcal A}_{\bf x}&=
\{(i,a,k) \mid 
(i,a) \in \Z_L \times [1,n], 
1 \le k \le (x^a_{i+1}-x^{a-1}_i)_+\}, \label{ykn1}\\
{\mathcal B}_{\bf x}&=
\{(i,a,k) \mid 
(i,a) \in \Z_L \times [1,n], 
1 \le k \le (x^a_{i}-x^{a+1}_{i+1})_+\}, \label{ykn2}
\end{align}
where the convention $x^0_i=x^{n+1}_i=0$ for all $i \in \Z_L$ 
should be applied throughout. 
\begin{example}\label{ex:yna}
For ${\bf x}$ in (\ref{hrk}) they read
\begin{align*}
{\mathcal A}_{\bf x} =\{
&(1,1,1), (1,1,2), (1,3,1), (2,1,1), (2,1,2), (2,1,3),\\
&(2,2,1), (2,2,2), (2,3,1), (3,1,1), (3,1,2),(3,1,3)\},\\
{\mathcal B}_{\bf x} =\{
&(1,1,1), (1,1,2), (1,3,1),(1,4,1),(1,4,2),(2,3,1),\\
&(3,1,1), (3,1,2), (3,2,1), (3,2,2), (3,2,3), (3,4,1)\}.
\end{align*}
\end{example}
\begin{example}\label{ex:tsk}
Take 
\begin{align*}
{\bf x} = \begin{pmatrix}
0 & 0 & 1\\
2 & 1 & 0\\
2 & 0 & 2\\
1 & 1 & 4
\end{pmatrix} \in B_6 \otimes B_4 \otimes B_3 \otimes B_1 = B(2,1,2,1)
\end{align*}
with $(n,L)=(4,3)$. Then we have
\begin{align*}
{\mathcal A}_{\bf x} &= \{
(1, 1, 1), (2, 1, 1), (2, 1, 2), (2, 1, 3), (2, 1, 4), (2, 2, 1), (3, 1, 1)\},\\
{\mathcal B}_{\bf x}&= \{
(1, 1, 1), (1, 2, 1), (1, 3, 1), (1, 3, 2), (3, 1, 1), (3, 1, 2), (3, 4, 1)\}.
\end{align*}
\end{example}
The coincidence of the cardinality of 
${\mathcal A}_{\bf x}$ and ${\mathcal B}_{\bf x}$ 
in these example is not accidental.

\begin{lemma}\label{le:air}
For any ${\bf x} \in B({\bf m})$, 
$\#{\mathcal A}_{\bf x} = \#{\mathcal B}_{\bf x}$ holds.
\end{lemma}
\begin{proof}
\begin{align*}
\#{\mathcal A}_{\bf x} -\#{\mathcal B}_{\bf x}&=
\sum_{i \in \Z_L}\sum_{a=1}^n(x^a_{i+1}-x^{a-1}_i)_+
-\sum_{i \in \Z_L}\sum_{a=1}^n(x^a_{i}-x^{a+1}_{i+1})_+\\
&=\sum_{i \in \Z_L}\sum_{a=0}^n
\Bigl( (x^{a+1}_{i+1}-x^{a}_{i})_+-
(x^a_{i}-x^{a+1}_{i+1})_+\Bigr)\\
&= \sum_{i \in \Z_L}\sum_{a=0}^n(x^{a+1}_{i+1}-x^{a}_{i})
=\sum_{i \in \Z_L}\sum_{a=0}^n(x^{a+1}_{i}-x^{a}_{i})
=\sum_{i \in \Z_L}(x^{n+1}_{i+1}-x^{0}_{i})=0.
\end{align*}
\end{proof}
We introduce the auxiliary sets 
${\mathcal A}$ and ${\mathcal B}$ by
\begin{align*}
{\mathcal A} &= \bigsqcup_{{\bf x} \in B({\bf m})} 
\{{\bf x}\}\times {\mathcal A}_{\bf x}
=\{\bigl({\bf x},(i,a,k)\bigr) \mid {\bf x} \in B({\bf m}), 
(i,a,k) \in {\mathcal A}_{\bf x}\},\\
{\mathcal B} &= \bigsqcup_{{\bf x} \in B({\bf m})} 
{\mathcal B}_{\bf x}\times \{{\bf x}\}
=\{\bigl((i,a,k),{\bf x}\bigr) \mid {\bf x} \in B({\bf m}), 
(i,a,k) \in {\mathcal B}_{\bf x}\}.
\end{align*}
The order of ${\bf x}$ and $(i,a,k)$ in the products
is a matter of convention but taking them differently as above
helps to distinguish ${\mathcal A}$ and ${\mathcal B}$.
By Lemma \ref{le:air} we know 
$\#{\mathcal A} = \#{\mathcal B}$.

\subsection{Bijection between ${\mathcal A}$ and ${\mathcal B}$}

Pick any $\bigl({\bf x},(i,a,k)\bigr) \in {\mathcal A}$.
It implies $1 \le k \le x^{a}_{i+1}-x^{a-1}_i$, hence
$x^{a-1}_i < x^{a}_{i+1}$ holds.
Suppose that
$x^{b-1}_i < x^{b}_{i+1}$ holds for $a \le b \le c$ until it
ceases to do so at $b=c+1$, namely $x^{c}_i \ge x^{c+1}_{i+1}$.
Due to the convention $x^{n+1}_i=0$, such $c \in [a,n]$ 
exists uniquely.
We define the map $T$ by
\begin{equation}\label{nzm}
\begin{split}
T: &\;{\mathcal A} \rightarrow {\mathcal B};
\;\;\bigl({\bf x},(i,a,k)\bigr) \mapsto \bigl((i,c,l),{\bf y}\bigr),\\
{\bf y}&= (y^{b}_j),\;\;
y^b_j = x^b_j\; \;\text{except}\;\;
(y^b_i,y^b_{i+1}) =  \begin{cases}
(x^b_i+x^{b}_{i+1} - x^{b-1}_{i}, x^{b-1}_i) & \text{if }  b \in [a+1,c],\\
(x^a_i+k, x^{a}_{i+1}-k) & \text{if }  b=a,
\end{cases}\\
l &= \begin{cases}k & \text{if } c=a,\\
x^{c}_{i+1}-x^{c-1}_{i} & \text{if }  c>a.
\end{cases}
\end{split}
\end{equation}
When $c=a$, the case $b \in [a+1,c]$ can just be omitted.
It is easy to see $T{\mathcal A} \subseteq {\mathcal B}$.
In fact from (\ref{ykn2}) and (\ref{nzm})
it suffices to check
$y^b_i \ge 0\; (b \in [a+1,c])$, $y^a_{i+1}\ge 0$ and 
$1 \le l \le y^c_i- y^{c+1}_{i+1}$. 
They all follow from the definition straightforwardly.

Similarly pick any $\bigl((i,a,k), {\bf x}\bigr) \in {\mathcal B}$.
It implies $1 \le k \le x^a_i - x^{a+1}_{i+1}$ hence 
$x^a_i > x^{a+1}_{i+1}$ holds.
Suppose that $x^b_i > x^{b+1}_{i+1}$ holds for 
$d \le b \le a$ until it ceases to do so at $b=d-1$, namely
$x^{d-1}_i \le x^d_{i+1}$.
Due to the convention $x^0_i = 0$, such $d \in [1, a]$ exists uniquely.
We define the map $S$ by
\begin{equation}\label{rna}
\begin{split}
S: &\;{\mathcal B} \rightarrow {\mathcal A};
\;\;\bigl((i,a,k), {\bf x}\bigr) \mapsto \bigl({\bf y}, (i,d,m)\bigr),\\
{\bf y}&= (y^{b}_j),\;\;
y^b_j = x^b_j\; \;\text{except}\;\;
(y^b_i,y^b_{i+1}) =  \begin{cases}
(x^a_i-k, x^{a}_{i+1}+k) & b=a,\\
(x^{b+1}_{i+1}, x^b_i+x^{b}_{i+1} - x^{b+1}_{i+1}) & b \in [d,a-1],
\end{cases}\\
m &= \begin{cases}
k & \text{if } d=a,\\
x^{d}_{i}- x^{d+1}_{i+1} & \text{if } d<a.
\end{cases}
\end{split}
\end{equation}
When $d=a$, the case $b \in [d,a-1]$ can just be omitted.
Again $S{\mathcal B} \subseteq {\mathcal A}$ is easily seen.

Focusing on the $i$-th and the $(i+1)$-th columns of the tableaux 
${\bf x}$ and ${\bf y}$, 
we set $(\alpha_b, \beta_b) = (x^b_i, x^b_{i+1})$ for simplicity.
Then their nontrivial changes in (\ref{nzm})  and (\ref{rna}) look as follows:
\begin{equation}\label{syr}
\begin{picture}(200,155)(80,-39)

\put(0,105){$\;\;x^b_i \;\;\; x^b_{i+1}$}
\multiput(-3,56)(0,18){2}{
\put(0,0){\line(1,0){52}}}
\put(24,51){\line(0,1){40}}

\put(0,80){$\alpha_{c+1}\;\;\,\beta_{c+1}$}
\put(0,62){$\,\;\alpha_{c}\;\;\;\;\;\;\beta_{c}$}

\multiput(-3,-5)(0,18){3}{
\put(0,0){\line(1,0){52}}}
\put(24,-23){\line(0,1){59}}\put(22.5,40){$\vdots$}

\put(0,20){$\alpha_{a+1}\;\;\beta_{a+1}$}
\put(0,1){$\,\;\alpha_{a}\;\;\;\;\;\,\beta_{a}$}
\put(0,-17){$\alpha_{a-1}\;\;\beta_{a-1}$}

\put(-3,-23){\line(0,1){115}}
\put(49,-23){\line(0,1){115}}

\put(18,73){\rotatebox{45}{\put(0,-4){$\ge$}}}
\put(18,52.5){\rotatebox{45}{\put(0,-4){$<$}}}
\put(18,28){\rotatebox{45}{\put(0,-4){$<$}}}
\put(18,10){\rotatebox{45}{\put(0,-4){$<$}}}
\put(18,-8){\rotatebox{45}{\put(0,-4){$<$}}}

\put(19,-37){(i)}

\put(65,37){$\overset{T}{\longmapsto}$}
% IMAGE of T

\put(100,0){
\put(0,105){$\;\;y^b_i \;\;\; y^b_{i+1}$}
\multiput(-3,56)(0,18){2}{
\put(0,0){\line(1,0){52}}}
\put(24,51){\line(0,1){40}}

\put(0,80){$\alpha_{c+1}\,\;\;\beta_{c+1}$}
\put(1,62){$\,\;\,\,\ast\;\;\;\;\,\alpha_{c-1}$}

\multiput(-3,-5)(0,18){3}{
\put(0,0){\line(1,0){52}}}
\put(24,-23){\line(0,1){59}}\put(22.5,40){$\vdots$}

\put(0,20){$\;\;\;\ast\;\;\;\;\;\;\alpha_{a}$}
\put(0,1){$\;\;\;\ast\;\;\;\,\, \beta_a\!\!-\!k$}
\put(0,-17){$\alpha_{a-1}\;\;\beta_{a-1}$}

\put(-3,-23){\line(0,1){115}}
\put(49,-23){\line(0,1){115}}
\put(18,-37){(ii)}
}

%%%%%% S%%%%%

\put(200,0){
\put(0,105){$\;\;x^b_i \;\;\; x^b_{i+1}$}

\multiput(-3,38)(0,18){3}{
\put(0,0){\line(1,0){52}}}
\put(24,33){\line(0,1){58}}

\put(0,80){$\alpha_{a+1}\;\;\beta_{a+1}$}
\put(0,62){$\,\;\alpha_{a}\;\;\;\;\;\;\beta_{a}$}
\put(0,44){$\alpha_{a-1}\;\;\beta_{a-1}$}

\multiput(-3,-5)(0,18){2}{
\put(0,0){\line(1,0){52}}}
\put(24,-23){\line(0,1){41}}\put(22.5,22){$\vdots$}

\put(0,1){$\,\;\alpha_{d}\;\;\;\;\;\,\beta_{d}$}
\put(0,-17){$\alpha_{d-1}\;\;\beta_{d-1}$}

\put(-3,-23){\line(0,1){115}}
\put(49,-23){\line(0,1){115}}

\put(20,73){\rotatebox{48}{\put(0,-4){$>$}}}
\put(20,55){\rotatebox{45}{\put(0,-4){$>$}}}
\put(20,37){\rotatebox{45}{\put(0,-4){$>$}}}
\put(20,12){\rotatebox{45}{\put(0,-4){$>$}}}
\put(19,-7){\rotatebox{45}{\put(0,-4){$\le$}}}

\put(17,-37){(iii)}

\put(65,37){$\overset{S}{\longmapsto}$}
% IMAGE of S
}

\put(300,0){
\put(0,105){$\;\;y^b_i \;\;\; y^b_{i+1}$}

\multiput(-3,38)(0,18){3}{
\put(0,0){\line(1,0){52}}}
\put(24,33){\line(0,1){58}}

\put(0,80){$\alpha_{a+1}\;\;\beta_{a+1}$}
\put(-1,62){$\alpha_{a}\!\!-\!k\;\;\;\;\ast$}
\put(0,44){$\;\;\beta_a\;\;\;\;\;\;\,\ast$}

\multiput(-3,-5)(0,18){2}{
\put(0,0){\line(1,0){52}}}
\put(24,-23){\line(0,1){41}}\put(22.5,22){$\vdots$}

\put(0,1){$\,\beta_{d+1}\;\;\;\;\,\ast$}
\put(0,-17){$\alpha_{d-1}\;\;\beta_{d-1}$}

\put(-3,-23){\line(0,1){115}}
\put(49,-23){\line(0,1){115}}
\put(17,-37){(iv)}
}

\end{picture}
\end{equation}
Here $\ast$ signifies the integers so chosen that the 
sum is conserved in each row.
It implies that the dots in the diagram like (\ref{hrk}) 
always move horizontally to the left (resp. right) neighbor slot 
by $T$ (resp. $S$).

\begin{proposition}\label{pr:tgm}
$T$ and $S$ are bijections satisfying 
$TS = \mathrm{id}_{\mathcal B}$ and $ST = \mathrm{id}_{\mathcal A}$.
\end{proposition}
\begin{proof}
Thanks to Lemma \ref{le:air} it suffices to prove 
$ST = \mathrm{id}_{\mathcal A}$.
Let $T: \bigl({\bf x}, (i,a,k) \bigr) \mapsto \bigl((i,c,l), {\bf y}\bigr)$.
The right hand side is depicted in (ii) of (\ref{syr}) and 
we have already seen that it belongs to ${\mathcal B}$ after (\ref{nzm}).
First we show that (ii) satisfies the inequalities in 
(iii)$_{(a,d)\rightarrow (c,a)}$.
In fact the bottom one $\beta_a-k \ge \alpha_{a-1}$,
for example, follows from the condition 
$(i,a,k) \in {\mathcal A}_{\bf x}$ (\ref{ykn1}).
The rest are similar.
Now we can inscribe the same numbers as 
(ii) into (iii)$_{(a,d)\rightarrow (c,a)}$.
Then the rule (iii) $\mapsto$ (iv) reproduces ${\bf x}$.
Second we check that $m$ given in (\ref{rna}) returns to  
the original value $k$ when applied to (ii).  For $d=a$ this is obvious.
For $d<a$ we have  
$m= (\mathrm{bottom}\, \ast \,\mathrm{ in \,(ii)}) - \alpha_a
= (\alpha_a+k)-\alpha_a=k$ as desired.
Thus $ST$ reproduces $(i,a,k)$ as well as ${\bf x}$. 
\end{proof}

\begin{example}\label{ex:mho}
For ${\bf x}$ and 
$(i,a,k) \in {\mathcal A}_{\bf x}$ in Example \ref{ex:tsk},
the image of $\bigl({\bf x}, (i,a,k) \bigr)$ under $T$ reads
{\small
\begin{align*}
&
\bigl({\bf x}, (1,1,1)\bigr) \mapsto 
\Bigl((1,1,1),
\begin{pmatrix}
 0 & 0 & 1 \\
 2 & 1 & 0 \\
 2 & 0 & 2 \\
 2 & 0 & 4
\end{pmatrix}\Bigr),\quad
\bigl({\bf x}, (2,1,1)\bigr) \mapsto 
\Bigl((2,2,1),
\begin{pmatrix}
 0 & 0 & 1 \\
 2 & 1 & 0 \\
 2 & 1 & 1 \\
 1 & 2 & 3
\end{pmatrix}\Bigr),\\
&
\bigl({\bf x}, (2,1,2)\bigr) \mapsto 
\Bigl((2,2,1),
\begin{pmatrix}
 0 & 0 & 1 \\
 2 & 1 & 0 \\
 2 & 1 & 1 \\
 1 & 3 & 2
\end{pmatrix}\Bigr),\quad
\bigl({\bf x}, (2,1,3)\bigr) \mapsto 
\Bigl((2,2,1),
\begin{pmatrix}
 0 & 0 & 1 \\
 2 & 1 & 0 \\
 2 & 1 & 1 \\
 1 & 4 & 1
\end{pmatrix}\Bigr),\\
&
\bigl({\bf x}, (2,1,4)\bigr) \mapsto 
\Bigl((2,2,1),
\begin{pmatrix}
0 & 0 & 1 \\
 2 & 1 & 0 \\
 2 & 1 & 1 \\
 1 & 5 & 0
\end{pmatrix}\Bigr),\quad
\bigl({\bf x}, (2,2,1)\bigr) \mapsto 
\Bigl((2,2,1),
\begin{pmatrix}
0 & 0 & 1 \\
 2 & 1 & 0 \\
 2 & 1 & 1 \\
 1 & 1 & 4
\end{pmatrix}\Bigr),\\
&
\bigl({\bf x}, (3,1,1)\bigr) \mapsto 
\Bigl((3,1,1),
\begin{pmatrix}
 0 & 0 & 1 \\
 2 & 1 & 0 \\
 2 & 0 & 2 \\
 0 & 1 & 5
\end{pmatrix}\Bigr).
\end{align*}
}
Reversing these arrows provides examples of the image under $S$.
\end{example}

\subsection{Stochastic dynamics}
In the transformation
$T: \bigl({\bf x},(i,a,k)\bigr) \mapsto \bigl((i,c,l),{\bf y}\bigr)$
in (\ref{nzm}), we write the 
uniquely determined element ${\bf y} \in B({\bf m})$
as ${\bf y} = T^k_{i,a}({\bf x})$.
In this way we have the deterministic evolution  
$T^k_{i,a}$ on $B({\bf m})$ such that 
\begin{align}\label{skb2}
T: \bigl({\bf x},(i,a,k)\bigr) \mapsto \bigl((i,c,l), T^k_{i,a}({\bf x})\bigr).
\end{align}
It is defined for each $(i,a,k) \in {\mathcal A}_{\bf x}$.
For a later convenience we also set 
\begin{equation}\label{skb:f}
T^k_{i,a}({\bf x}) = {\bf x} \quad \text{if}\; \;(i,a,k) \not\in 
{\mathcal A}_{\bf x}.
\end{equation}

Given ${\bf x} = (x^a_i) \in B({\bf m})$,
suppose that $t := x^a_{i+1}-x^{a-1}_i \ge 1$ 
so that
${\mathcal A}_{\bf x}$ (\ref{ykn1}) contains 
$(i,a,1), (i,a,2), \ldots, (i,a, t)$.
Then it is easy to see
$T^k_{i,a}({\bf x}) = (T^1_{i,a})^k({\bf x})$
for $k \in [1,t]$.

By the $n$-line process with prescribed multiplicity ${\bf m}$,
we mean the stochastic process on $B({\bf m})$ in which each state
${\bf x} \in B({\bf m})$ undergoes the time evolution 
$T^k_{i,a}$ with an equal transition rate for all 
$(i,a,k) \in {\mathcal A}_{\bf x}$.  
Let 
$\bigoplus_{{\bf x} \in B({\bf m})} \C|{\bf x}\rangle$ be its space of states 
having the basis $\{|{\bf x}\rangle\}$ labeled with 
${\bf x}=(x^a_i) \in B({\bf m})$ (\ref{mkr:gk}).
By the definition, the Markov matrix of the $n$-line process (LP) reads as
\begin{align}\label{skb:t}
H_{\mathrm{LP}}= \sum_{i\in \Z_L}\sum_{a=1}^n\sum_{k \ge 1}
(T^k_{i,a}-1),
\end{align}
which is convergent owing to (\ref{skb:f}).

\subsection{Steady state}

The $n$-line process on $B({\bf m})$ possesses a unique steady state.
Let $\mu({\bf x})$ denote its probability distribution. 

\begin{theorem}\label{th:akn}
The stationary probability distribution of the $n$-line process is uniform.
Namely $\mu({\bf x})$ is independent of ${\bf x} \in B({\bf m})$.
\end{theorem}
\begin{proof}
Since the steady state is unique, it suffices to show that 
the total rate $P_\text{in}$ 
jumping into a given state ${\bf x}$ is equal to 
the total rate $P_\text{out}$ jumping out of it under the uniform choice
$\mu({\bf x}) = \mu$.
One has $P_\text{out} = (\#{\mathcal A}_{\bf x})\mu({\bf x})
= (\#{\mathcal A}_{\bf x})\mu$.
One the other hand, $P_\text{in}$ is calculated as
\begin{align*}
P_\text{in} &= \sum_{{\bf y} \in B({\bf m})}
\sum_{(i,a,k) \in {\mathcal A}_{\bf y}}
\theta(T^k_{i,a}({\bf y}) = {\bf x})\mu({\bf y})\\
&= \mu\sum_{({\bf y}, (i,a,k)) \in {\mathcal A}}
\sum_{(j,b,l) \in {\mathcal B}_{\bf x}}
\theta\left(T: \bigl({\bf y},(i,a,k) \bigr) \mapsto \bigl((j,b,l),{\bf x}\bigr)\right)\\
&= \mu
\sum_{(j,b,l) \in {\mathcal B}_{\bf x}}\sum_{({\bf y}, (i,a,k)) \in {\mathcal A}}
\theta\left(S: \bigl((j,b,l),{\bf x}\bigr) \mapsto \bigl({\bf y},(i,a,k) \bigr)\right)\\
&= \mu \sum_{(j,b,l) \in {\mathcal B}_{\bf x}}1 = (\#{\mathcal B}_{\bf x})\mu
=(\#{\mathcal A}_{\bf x})\mu,
\end{align*}
where the last equality is due to Lemma \ref{le:air}. 
\end{proof}

We summarize the result as
\begin{corollary}[Steady state of $n$-line process]\label{akn:u}
\begin{align*}
H_{\mathrm{LP}}|\Omega({\bf m})\rangle = 0,\qquad
|\Omega({\bf m})\rangle = \sum_{{\bf x} \in B({\bf m})}|{\bf x}\rangle.
\end{align*}
\end{corollary}

\section{Projection from $n$-line process to $n$-TAZRP}\label{sec:pr}

\subsection{Combinatorial $R$}\label{ss:cr}

The $B_\ell$ defined in (\ref{mkr}) is a labeling set of a basis of  
the $\ell$-fold symmetric tensor representation of the 
quantum affine algebra $U_q(\widehat{sl}_L)$ \cite{D86,J}.
The {\em combinatorial $R$} is the quantum $R$ matrix at $q=0$.
It is a bijection 
\begin{align}\label{obt}
R=R_{\ell,m}:\;B_\ell \otimes B_m \rightarrow B_m \otimes B_\ell; \quad
{\bf x} \otimes {\bf y} \mapsto {\bf y}' \otimes {\bf x}'.
\end{align}
To describe it explicitly we set ${\bf x}=(x_1,\ldots, x_L), {\bf y}=(y_1,\ldots, y_L), 
{\bf x'}=(x'_1,\ldots, x'_L)$ and 
${\bf y'}=(y'_1,\ldots, y'_L)$.
Note that $\sum_i x_i = \sum_i x'_i = \ell$ and 
$\sum_i y_i = \sum_i y'_i = m$.
Given ${\bf x}=(0,3,2,2,1)$ and ${\bf y}=(2,0,2,1,0)$ for instance,
we depict them as the diagram (i) given below.
In this paper we will only encounter the case $\ell > m$.
Then the algorithm \cite[Rule 3.11]{NY} known as the {\em NY-rule} for finding 
${\bf x}'$ and ${\bf y}'$ goes as follows: 

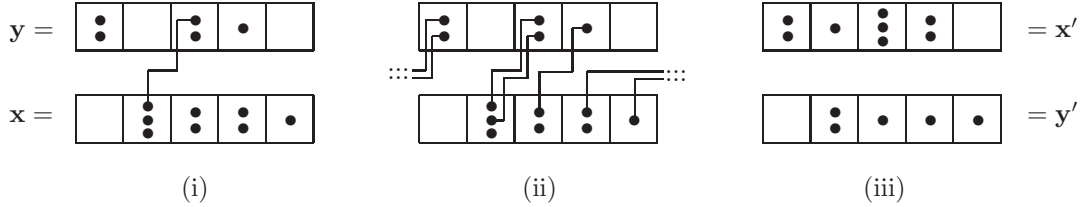
\begin{picture}(380,85)(-65,-22)

%%%%%%%%%% (i) 
\put(-40,35){
\put(-25,-27){${\bf x}=$}\put(-25,5){${\bf y}=$}
\multiput(0,0)(18,0){6}{\put(0,0){\line(0,1){18}}}
\put(0,0){\line(1,0){90}}\put(0,18){\line(1,0){90}}
\put(7,9){$\bullet$}\put(7,3){$\bullet$}
\put(43,9){$\bullet$}\put(43,3){$\bullet$}
\put(61,6){$\bullet$}
\put(45,11.5){\line(-1,0){6.5}}\put(38.5,11.5){\line(0,-1){19}}
\put(38.5,-7.5){\line(-1,0){11}}\put(27.5,-7.5){\line(0,-1){14}}
}

\put(-40,0){
\multiput(0,0)(18,0){6}{\put(0,0){\line(0,1){18}}}
\put(0,0){\line(1,0){90}}\put(0,18){\line(1,0){90}}
\put(25,11.5){$\bullet$}\put(25,6.2){$\bullet$}\put(25,1){$\bullet$}
\put(43,9){$\bullet$}\put(43,3){$\bullet$}
\put(61,9){$\bullet$}\put(61,3){$\bullet$}
\put(79,6){$\bullet$}}

%%%%%%%%%% (ii)
\put(130,0){
\put(-40,35){
\multiput(0,0)(18,0){6}{\put(0,0){\line(0,1){18}}}
\put(0,0){\line(1,0){90}}\put(0,18){\line(1,0){90}}
\put(7,9){$\bullet$}\put(7,3){$\bullet$}
\put(43,9){$\bullet$}\put(43,3){$\bullet$}
\put(61,6){$\bullet$}
%-
\put(9,11.5){\line(-1,0){6.5}}\put(2.5,11.5){\line(0,-1){19}}\put(2.5,-7.5){\line(-1,0){5}}
\put(-12,-7.9){...}
\put(9,5.5){\line(-1,0){4}}\put(5,5.5){\line(0,-1){16}}\put(5,-10.5){\line(-1,0){7.5}}
\put(-12,-10.9){...}
%--
\put(45,11.5){\line(-1,0){6.5}}\put(38.5,11.5){\line(0,-1){19}}
\put(38.5,-7.5){\line(-1,0){11}}\put(27.5,-7.5){\line(0,-1){14}}
\put(45,5.5){\line(-1,0){4}}\put(41,5.5){\line(0,-1){16}}
\put(41,-10.5){\line(-1,0){9}}\put(32,-10.5){\line(0,-1){16}}
\put(32,-26.5){\line(-1,0){4}}
%---
\put(63,8.5){\line(-1,0){5}}\put(58,8.5){\line(0,-1){16.5}}\put(58,-8){\line(-1,0){12.5}}
\put(45.5,-8){\line(0,-1){15}}
%----
\put(92.5,-7.9){\line(-1,0){29}}\put(63.5,-7.9){\line(0,-1){15}}\put(93,-8.5){...}
\put(92.5,-10.9){\line(-1,0){11}}\put(81.5,-10.9){\line(0,-1){16}}\put(93,-11.5){...}
}

\put(-40,0){
\multiput(0,0)(18,0){6}{\put(0,0){\line(0,1){18}}}
\put(0,0){\line(1,0){90}}\put(0,18){\line(1,0){90}}
\put(25,11.5){$\bullet$}\put(25,6.2){$\bullet$}\put(25,1){$\bullet$}
\put(43,9){$\bullet$}\put(43,3){$\bullet$}
\put(61,9){$\bullet$}\put(61,3){$\bullet$}
\put(79,6){$\bullet$}}}

%%%%%%%%%%(iii) %%%%%%%%%%%%%%

\put(260,0){
\put(-40,35){
\put(100,5){$={\bf x}'$}\put(100,-27){$={\bf y}'$}
\multiput(0,0)(18,0){6}{\put(0,0){\line(0,1){18}}}
\put(0,0){\line(1,0){90}}\put(0,18){\line(1,0){90}}
\put(7,9){$\bullet$}\put(7,3){$\bullet$}
\put(25,6){$\bullet$}
\put(43,11.5){$\bullet$}\put(43,6.5){$\bullet$}\put(43,1){$\bullet$}
\put(61,9){$\bullet$}\put(61,3){$\bullet$}}

\put(-40,0){
\multiput(0,0)(18,0){6}{\put(0,0){\line(0,1){18}}}
\put(0,0){\line(1,0){90}}\put(0,18){\line(1,0){90}}
\put(25,9){$\bullet$}\put(25,3){$\bullet$}
\put(43,6){$\bullet$}
\put(61,6){$\bullet$}
\put(79,6){$\bullet$}}}

\put(-23,-6){\put(23,-14){(i)}\put(152,-14){(ii)}\put(281,-14){(iii)}}
\end{picture}

\renewcommand{\labelenumi}{(\roman{enumi})}

\begin{enumerate}
\item Choose a dot, say $d$, in ${\bf y}$ and connect it to a 
dot $d'$ in ${\bf x}$ to form a pair.
$d'$ should be the rightmost one 
among those located strictly left of $d$. 
If there is no such dot, take $d'$ to be the rightmost one in ${\bf x}$,
i.e., seek such $d'$ under the {\em periodic boundary condition}.
(Dots in the same box are regarded to be in the same position.
If $d'$ is to be chosen from more than one dots in a box, 
any of them can be taken.)

\item Repeat (i) for yet unpaired dots until all dots in ${\bf y}$ 
are paired to some dots in ${\bf x}$.

\item Move the $l-m$ unpaired dots in ${\bf x}$ vertically up to ${\bf y}$.
The resulting diagrams give ${\bf x}' $ and ${\bf y}'$.
\end{enumerate}
\renewcommand{\labelenumi}{(\arabic{enumi})}

In the above example we have ${\bf x}' = (2,1,3,2,0)$ and
${\bf y}' =(0,2,1,1,1)$.
The lines pairing the dots are called {\em $H$-lines}\footnote{
The nomenclature $H$ originates in \cite[(2.5)]{DJKMO} and \cite[(1.1.3)]{KMN} etc., 
which has the meaning of {\em local Hamiltonian} 
of corner transfer matrices \cite{Bax}.}.
We note that the NY-rule implies 
\begin{align}\label{mri}
R({\bf x}\otimes {\bf y}) = {\bf y}' \otimes {\bf x}' 
\in B_m \otimes B_\ell\;\;\Longrightarrow 
\;{\bf x} \ge {\bf y}', \;{\bf x}' \ge {\bf y}
\end{align}
for $\ell > m$, 
where in general 
for ${\bf u}=(u_1,\ldots, u_L), {\bf v}=(v_1, \ldots, v_L) \in \Z^L$,
${\bf u} \ge {\bf v}$ is defined by 
${\bf u} - {\bf v} \in (\Z_{\ge 0})^L$.

\begin{remark}\label{re:mri}\par \noindent
\begin{enumerate}
\item
In (i) and (ii), the $H$-lines depend on the order of choosing dots from 
${\bf y}$. However the final result ${\bf x}'$ and ${\bf y}'$ are independent of it 
due to \cite[Prop.3.20]{NY}. 

\item
A similar algorithm is known for the case $\ell<m$.
$R_{\ell,m}$ is identity for $\ell=m$.
In any case $R_{\ell,m}R_{m,\ell} = \mathrm{id}_{B_m \otimes B_\ell}$ holds.

\item
The above rule is close to but slightly different from the combinatorial $R$ of the 
{\em antisymmetric} tensor representation 
which was identified \cite{KMO} 
with the arrival/service/departure process relevant to TASEP \cite{FM}.
In this interpretation one regards that time increases horizontally to the left
(apart from the periodic boundary condition).
Then the customers (dots in ${\bf y}$) in the present case 
have to avoid the service (dots in ${\bf x}$) 
that becomes available simultaneously with their arrival. 
Another significant difference is, there can be {\em multiple} arrival and service at a time reflecting the symmetric tensor representation.

\item The following explicit formula having background in 
geometric crystals and soliton cellular automata \cite{IKT}  is known to hold 
either for $\ell \ge m$ or $\ell \le m$ \cite{Y}:
\begin{align*}
&x'_i = x_i+Q_i(x,y)-Q_{i-1}(x,y),\quad 
y'_i = y_i+Q_{i-1}(x,y)-Q_i(x,y)\quad (i \in \Z_L),\\
&Q_i(x,y) = \min \Bigl\{ 
\sum_{j=1}^{k-1}x_{i+j} + \sum_{j=k+1}^L y_{i+j}
\mid 1 \le k \le L \Bigr\}.
\end{align*}
\end{enumerate}
\end{remark}

It is customary to depict the relation (\ref{obt}) as a vertex:
\begin{equation}\label{ask3}
\thicklines
\begin{picture}(50,50)(-20,-20)

\put(-20,-2){${\bf x}$}
\put(15,-2){${\bf x}'$}
\put(-3,-18){${\bf y}$}
\put(-3,15.5){${\bf y}'$}
\put(-10,0){\vector(1,0){20}}
\put(0,-10){\vector(0,1){20}}

\end{picture}
\end{equation}
One may rotate it arbitrarily.
The thick arrows here carry crystals $B_\ell$ or $B_m$.
They are to be distinguished from the thin arrows 
carrying a Fock space $F$ which will be used in 
Section \ref{ss:fcr} and \ref{ss:mp}.

The most significant property of the combinatorial $R$ is 
the Yang-Baxter equation \cite{Bax}:
\begin{equation*}
(R\otimes 1)(1\otimes R)(R\otimes 1) = (1\otimes R)(R\otimes 1) (1\otimes R)
\end{equation*}
as maps $B_\ell \otimes B_m \otimes B_k \rightarrow
B_k \otimes B_m \otimes B_\ell$ for any $\ell, m,k$.
For example the two sides acting on the element 
$0121\otimes 1101 \otimes 2000 \in B_4 \otimes B_3 \otimes B_2$ 
lead to\footnote{
$(0,1,2,1)$ for example is denoted by 0121 for simplicity. 
A similar convention will also be used in the rest of the paper.} 
\begin{equation}\label{ask:dk}
\thicklines
\begin{picture}(200,90)(0,0)

\put(0,75){\put(0,0){0011} \put(30,0){0111} \put(60,0){3100}}

\put(0,50){
\put(13.5,9){\vector(3,2){20}}\put(33.5,9){\vector(-3,2){20}}
\put(70,9){\vector(0,1){12}}
}

\put(0,50){\put(0,0){0021} \put(30,0){0101} \put(60,0){3100}}

\put(0,25){
\put(10,9){\vector(0,1){12}}\put(43.5,9){\vector(3,2){20}}\put(63.5,9){\vector(-3,2){20}}
}

\put(0,25){\put(0,0){0021} \put(30,0){1201} \put(60,0){2000}}

\put(13.5,9){\vector(3,2){20}}\put(33.5,9){\vector(-3,2){20}}\put(70,9){\vector(0,1){12}}

\put(0,0){0121} \put(30,0){1101}\put(60,0){2000}

\put(95,40){$=$}

\put(120,0){
\put(0,75){\put(0,0){0011} \put(30,0){0111} \put(60,0){3100}}

\put(0,50){
\put(10,9){\vector(0,1){12}}\put(43.5,9){\vector(3,2){20}}\put(63.5,9){\vector(-3,2){20}}
}

\put(0,50){\put(0,0){0011} \put(30,0){0211} \put(60,0){3000}}

\put(0,25){
\put(13.5,9){\vector(3,2){20}}\put(33.5,9){\vector(-3,2){20}}\put(70,9){\vector(0,1){12}}
}

\put(0,25){\put(0,0){0121} \put(30,0){0101} \put(60,0){3000}}

\put(10,9){\vector(0,1){12}}\put(43.5,9){\vector(3,2){20}}\put(63.5,9){\vector(-3,2){20}}

\put(0,0){0121} \put(30,0){1101}\put(60,0){2000}
}

\end{picture}
\end{equation}
Here $=$ means that starting from the same bottom line 
one ends up with the same top line.
In this way one can move the arrows across other vertices without 
changing outer states. 
This property will be utilized efficiently in the sequel. 
 
Combinatorial $R$'s form the most systematic examples of 
the set theoretical solutions to the Yang-Baxter equation \cite{D92,V} 
connected to the representation theory of quantum groups, which have 
numerous applications \cite{HKOTT,IKT,KMN, Ku,NY,Y}.

\subsection{$B_+({\bf m})$ and elementary bijection $\varphi$}
Recall that the sets of configurations are given by 
$B({\bf m})$ (\ref{mkr}) for $n$-line process  
and by $S({\bf m})$ (\ref{kgc}) for $n$-TAZRP.
We introduce a subsidiary set 
\begin{align*}
B_+({\bf m}) = \{{\bf x}^1\otimes \cdots \otimes {\bf x}^n 
\in B({\bf m})\mid
{\bf x}^1 \ge \cdots \ge {\bf x}^n \} \subseteq B({\bf m}),
\end{align*}
where $\ge$ is defined after (\ref{mri}).
There is an elementary bijection between 
$B_+({\bf m})$ and $S({\bf m})$ as 
\begin{equation}\label{sra}
\begin{split}
\varphi: \; S({\bf m})&\overset{\sim}{\longrightarrow} B_+({\bf m});
\quad
{\boldsymbol \sigma} = (\sigma^{a}_i) \longmapsto
\varphi^1({\boldsymbol \sigma})\otimes \cdots \otimes 
\varphi^n({\boldsymbol \sigma}),\\
\varphi^a({\boldsymbol \sigma}) &
=(\sigma^{a}_1+\sigma^{a+1}_1+ \cdots + \sigma^{n}_1, \ldots,
\sigma^{a}_L+\sigma^{a+1}_L+ \cdots + \sigma^{n}_L),
\end{split}
\end{equation}
where 
$\varphi^a({\boldsymbol \sigma})\in B_{\ell_a}$ and 
$\varphi^1({\boldsymbol \sigma}) \ge \cdots 
\ge \varphi^n({\boldsymbol \sigma})$ are obvious.

\begin{example}\label{ex:ask:g}
The $\varphi^1({\boldsymbol \sigma}),\ldots, \varphi^n({\boldsymbol \sigma})$ 
are easily read off from
the dot diagram (see (\ref{hrk})) by regarding 
a particle of species $a$ as a column of $a$ dots filled from the bottom.
The following is an example for $(n,L)=(3,4)$, giving 
$\varphi^1({\boldsymbol \sigma}) = 1213\in B_7$,
$\varphi^2({\boldsymbol \sigma}) = 1201\in B_4$ and 
$\varphi^3({\boldsymbol \sigma}) = 0101\in B_2$.
\begin{equation*}
\begin{picture}(190,70)(-175,-10)

\put(-260,0){
\put(25,45){mutiplicity rep.} \put(117,45){multiset rep.}
\put(0,30){${\boldsymbol \sigma}=(010,011,100,201) = 
(2,23,1,113)$}
\put(0,15){$\varphi({\boldsymbol \sigma}) = 1213 \otimes 1201 \otimes 0101$}
}

\multiput(0.5,0)(0,20){4}{\put(0,0){\line(1,0){80}}}
\multiput(0.5,0)(20,0){5}{\put(0,0){\line(0,1){60}}}

\put(-45,47){$\varphi^3({\boldsymbol \sigma})=$}
\put(-45,27){$\varphi^2({\boldsymbol \sigma})=$}
\put(-45,7){$\varphi^1({\boldsymbol \sigma})=$}

\put(28,47){$\bullet$}\put(68,47){$\bullet$}

\put(8,27){$\bullet$}
\put(25,27){$\bullet$}\put(32,27){$\bullet$}
\put(68,27){$\bullet$}

\put(8,7){$\bullet$}
\put(25,7){$\bullet$}\put(32,7){$\bullet$}
\put(48,7){$\bullet$}
\put(62.5,7){$\bullet$}\put(68,7){$\bullet$}\put(73.5,7){$\bullet$}

\put(8,-11){2} \put(26,-11){23} \put(48,-11){1} \put(63,-11){113}
\end{picture}
\end{equation*}
\end{example}

The inverse of $\varphi$ is given by
\begin{equation}\label{ask2}
\begin{split}
\varphi^{-1}: B_+({\bf m})
& \;\overset{\sim}{\longrightarrow}  S({\bf m}); \quad
{\bf x}=(x^a_i) \longmapsto (\sigma^{a}_i),\\
\sigma^{a}_i & = x^a_i-x^{a+1}_i\;\;(x^{n+1}_i=0).
\end{split}
\end{equation}
See Section \ref{ss:nms}.
The maps $\varphi^{\pm 1}$ are also simply seen 
by the following diagram for each $i \in \Z_L$:
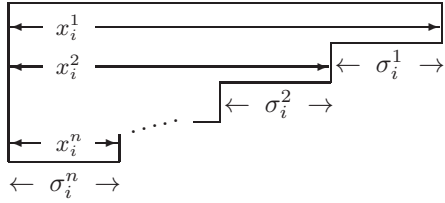
\begin{equation*}
\begin{picture}(190,85)(-15,-5)

\put(13, 61){\vector(-1,0){12}}\put(33, 61){\vector(1,0){130}}
\put(18, 58){{\small $x^{1}_i$}}

\put(13, 46){\vector(-1,0){12}}\put(33, 46){\vector(1,0){88}}
\put(18, 43){{\small $x^{2}_i$}}

\put(12, 17){\vector(-1,0){11}}\put(33, 17){\vector(1,0){8}}
\put(18, 14){{\small $x^{n}_i$}}

\put(0,70){\line(1,0){164}}
\put(164,70){\line(0,-1){15}}
\put(122,55){\line(0,-1){15}}
\put(80,40){\line(0,-1){15}}
\put(80,25){\line(-1,0){10}}

\put(122,45){
\put(0,10){\line(1,0){42}}\put(1.5,-1){$\leftarrow \;\sigma^{1}_i \;\rightarrow$} 
}

\put(80,30){
\put(0,10){\line(1,0){42}}\put(1.5,-1){$\leftarrow \;\sigma^{2}_i \;\rightarrow$} 
}

\multiput(46,21)(4,1){5}{.}

\put(0,10){\line(0,1){60}}\put(42,10){\line(0,1){10}}
\put(0,10){\line(1,0){42}}\put(0,-1){$\leftarrow \; \sigma^{n}_i \;\rightarrow$} 

\end{picture}
\end{equation*}
In this Young diagram for site $i$, a particle with species $a$
corresponds to a column of depth $a$. 

\subsection{Projection $\pi$}
We are going to construct a map 
$\pi: B({\bf m}) \rightarrow S({\bf m})$ as the composition of maps
\begin{alignat}{3}
&B({\bf m}) & \quad\overset{\pi^1\otimes \cdots \otimes \pi^n}{\longrightarrow} 
&\qquad\quad \;B_+({\bf m}) & \overset{\varphi^{-1}}{\longrightarrow} 
&\qquad \quad \;S({\bf m})\nonumber\\
&\;\;\;\;{\bf x}&\;\, \longmapsto \quad&
\;\pi^1({\bf x}) \otimes \cdots \otimes \pi^n({\bf x})
\;\;\;&\longmapsto & \;\;{\boldsymbol \sigma} = (\sigma_1, \ldots, \sigma_L),
\label{sae}
\end{alignat}
where $\varphi^{-1}$ is given by (\ref{ask2}).
Thus the remaining task is to construct 
\begin{align*}
\pi^a: B({\bf m})=B_{\ell_1} \otimes \cdots \otimes B_{\ell_n} 
\rightarrow B_{\ell_a};\quad
{\bf x}={\bf x}^1\otimes \cdots \otimes {\bf x}^n \mapsto 
\pi^a({\bf x})
\end{align*}
for each $a \in [1,n]$.
For $a=1$ we set $\pi^1({\bf x})={\bf x}^1$.
For $a \in [2,n]$ we set
\begin{equation}\label{srn}
R^1\cdots R^{a-2}R^{a-1}(
{\bf x}^1\otimes \cdots \otimes {\bf x}^a)
=\pi^a({\bf x}) \otimes {\bf u}^1 \otimes \cdots \otimes {\bf u}^{a-1},
\end{equation}
where $R^b$ is the combinatorial $R$ acting on the $(b,b+1)$-th components from the left.
The product $R^1\cdots R^{a-1}$ lets $B_{\ell_a}$ penetrate through
$B_{\ell_1} \otimes \cdots \otimes B_{\ell_{a-1}}$
bringing it to the left end of the tensor product.
The ${\bf u}^1 \otimes \cdots \otimes {\bf u}^{a-1}
\in B_{\ell_1} \otimes \cdots B_{\ell_{a-1}}$ 
denotes the element thereby produced, which will not be used in the sequel.
The $\pi^a({\bf x})$ is depicted for $a=1,2,3$ as
\begin{equation}\label{kan}
\begin{picture}(200,50)(-80,-20)
\thicklines

\put(-118,-2){$\pi^1({\bf x})$}
\put(-80,0){\vector(-1,0){10}}
\put(-75,-2){${\bf x}^1$}

\put(-40,-2){$\pi^2({\bf x})$}
\put(13,-2){${\bf x}^2$}
\put(-3,-20){${\bf x}^1$}
\put(-4,15.5){${\bf u}^1$}
\put(10,0){\vector(-1,0){23}}
\put(0,-10){\vector(0,1){21}}

\put(90,0){

\put(-40,-2){$\pi^3({\bf x})$}
\put(25,0){\vector(-1,0){38}}
\put(-3,-20){${\bf x}^1$}\put(-4,15.5){${\bf u}^1$}
\put(0,-10){\vector(0,1){21}}
\put(12,-20){${\bf x}^2$}\put(11,15.5){${\bf u}^2$}
\put(15,-10){\vector(0,1){21}}\put(28,-2){${\bf x}^3$}

}

\end{picture}
\end{equation}
The diagram for $\pi^3({\bf x})$ is a 
concatenation of (\ref{ask3}) rotated by 90 degrees.
Note that $\pi^a({\bf x})$ actually depends only on
the left $a$ components of 
${\bf x}={\bf x}^1 \otimes \cdots \otimes {\bf x}^n$.
We postpone the proof of the fact 
$\pi^1({\bf x}) \otimes \cdots \otimes \pi^n({\bf x}) 
\in B_+({\bf m})$ to Lemma \ref{le:sae}.

\begin{example}\label{ex:srn}
Consider 
${\bf x}=114\otimes 202 \otimes 210 \otimes 001$ in Example \ref{ex:tsk}.
Then $\pi^1({\bf x})=114$ and the other $\pi^a({\bf x})$'s are determined 
from the diagrams 
\begin{equation*}
\begin{picture}(380,50)(-30,-25)
\thicklines

\put(-32,-3){$112$}
\put(13,-3){$202$}
\put(-7,-20){$114$}
\put(-7,15){$204$}
\put(10,0){\vector(-1,0){23}}
\put(0,-10){\vector(0,1){21}}

\put(100,0){
\put(-32,-3){$111$}
\put(13,-3){$102$}
\put(-7,-20){$114$}
\put(-7,15){$105$}
\put(10,0){\vector(-1,0){23}}
\put(0,-10){\vector(0,1){21}}
\put(43,0){
\put(13,-3){$210$}
\put(-7,-20){$202$}
\put(-7,15){$310$}
\put(10,0){\vector(-1,0){23}}
\put(0,-10){\vector(0,1){21}}}}

\put(240,0){
\put(-32,-3){$001$}
\put(13,-3){$100$}
\put(-7,-20){$114$}
\put(-7,15){$213$}
\put(10,0){\vector(-1,0){23}}
\put(0,-10){\vector(0,1){21}}
\put(43,0){
\put(13,-3){$010$}
\put(-7,-20){$202$}
\put(-7,15){$112$}
\put(10,0){\vector(-1,0){23}}
\put(0,-10){\vector(0,1){21}}
\put(43,0){
\put(13,-3){$001$}
\put(-7,-20){$210$}
\put(-7,15){$201$}
\put(10,0){\vector(-1,0){23}}
\put(0,-10){\vector(0,1){21}}}}}

\end{picture}
\end{equation*}
as $\pi^2({\bf x})=112$, $\pi^3({\bf x})=111$
and $\pi^4({\bf x})=001$. 
They indeed satisfy 
$\pi^1({\bf x}) \ge\pi^2({\bf x}) \ge\pi^3({\bf x}) \ge\pi^4({\bf x})$,
assuring $\pi^1({\bf x}) \otimes\pi^2({\bf x}) \otimes\pi^3({\bf x}) \otimes\pi^4({\bf x}) 
\in B_+({\bf m})$. 
Applying the bijection $\varphi^{-1}$ (\ref{ask2}) further we obtain 
the $4$-TAZRP state
$\pi({\bf x}) = (3,3,1124)$ in multiset representation and 
$(0010,0010,2101)$ in multiplicity representation.
\end{example}

One can construct a {\em single} diagram 
that depicts $\pi^1({\bf x}), \ldots, \pi^n({\bf x})$ simultaneously.
We illustrate the procedure for $n=3$.
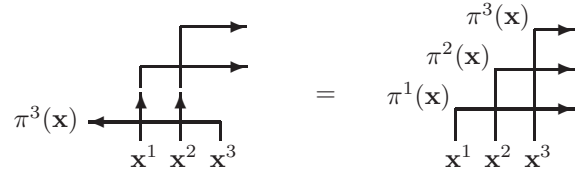
\begin{equation*}
\begin{picture}(380,65)(-30,-25)
\thicklines

\put(105,0){
\put(0,-15){\vector(0,1){18}}\put(15,-15){\vector(0,1){18}}

\put(0,5){\line(0,1){8}}\put(0,13){\vector(1,0){40}}
\put(15,5){\line(0,1){23}}\put(15,28){\vector(1,0){25}}

\put(30,-15){\line(0,1){7}}\put(30,-8){\vector(-1,0){50}}
\put(-4,-25){${\bf x}^1$}\put(11,-25){${\bf x}^2$}\put(26,-25){${\bf x}^3$}
\put(-48,-10){$\pi^3({\bf x})$}
}

\put(171,0){$=$}
\put(224,0){
\put(0,-15){\line(0,1){12}}\put(0,-3){\vector(1,0){45}}
\put(15,-15){\line(0,1){27}}\put(15,12){\vector(1,0){30}}
\put(30,-15){\line(0,1){42}}\put(30,27){\vector(1,0){15}}
\put(-4,-25){${\bf x}^1$}\put(11,-25){${\bf x}^2$}\put(26,-25){${\bf x}^3$}
\put(4,29){$\pi^3({\bf x})$}
\put(-11,14){$\pi^2({\bf x})$}\put(-26,-1){$\pi^1({\bf x})$}
}

\end{picture}
\end{equation*}
First we attach an extra vertex on top of the defining diagram of $\pi^3({\bf x})$.
This enables us to apply the Yang-Baxter equation to move the arrow going 
from ${\bf x}^3$ to $\pi^3({\bf x})$ 
around to get the right diagram, where the newly generated states on the diagonal
boundary edges are identified with 
$\pi^1({\bf x})$ and $\pi^2({\bf x})$ by (\ref{kan}).
A similar procedure for $n=4$ goes as
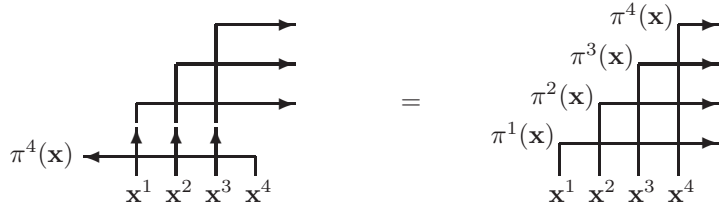
\begin{equation*}
\begin{picture}(380,85)(-80,-30)
\thicklines

\put(30,5){\line(0,1){37}}\put(30,42){\vector(1,0){30}}
\put(15,5){\line(0,1){22}}\put(15,27){\vector(1,0){45}}
\put(0,5){\line(0,1){7}}\put(0,12){\vector(1,0){60}}
\put(0,-15){\vector(0,1){18}}\put(15,-15){\vector(0,1){18}}\put(30,-15){\vector(0,1){18}}
\put(45,-15){\line(0,1){7}}\put(45,-8){\vector(-1,0){65}}

\put(-4,-25){${\bf x}^1$}\put(11,-25){${\bf x}^2$}\put(26,-25){${\bf x}^3$}\put(41,-25){${\bf x}^4$}
\put(-48,-10){$\pi^4({\bf x})$}

\put(100,10){$=$}
\put(160,0){
\put(0,-15){\line(0,1){12}}\put(0,-3){\vector(1,0){60}}
\put(15,-15){\line(0,1){27}}\put(15,12){\vector(1,0){45}}
\put(30,-15){\line(0,1){42}}\put(30,27){\vector(1,0){30}}
\put(45,-15){\line(0,1){57}}\put(45,42){\vector(1,0){15}}
\put(-4,-25){${\bf x}^1$}\put(11,-25){${\bf x}^2$}\put(26,-25){${\bf x}^3$}\put(41,-25){${\bf x}^4$}
\put(0,-2){\put(19,44){$\pi^4({\bf x})$}\put(4,29){$\pi^3({\bf x})$}
\put(-11,14){$\pi^2({\bf x})$}\put(-26,-1){$\pi^1({\bf x})$}}
}

\end{picture}
\end{equation*}
Here we have used the diagram representation of $\pi^3({\bf x})$ derived for $n=3$.
For $n$ general it look as
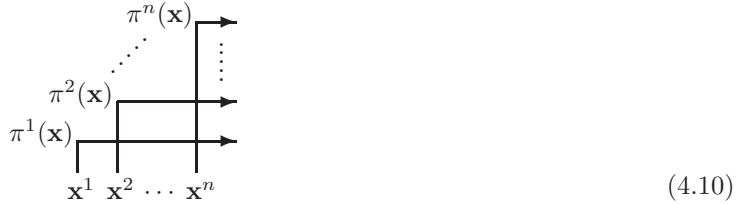
\begin{equation}\label{kan:t}
\begin{picture}(150,75)(-40,-23)
\thicklines

\put(0,-15){\line(0,1){12}}\put(0,-3){\vector(1,0){60}}
\put(15,-15){\line(0,1){27}}\put(15,12){\vector(1,0){45}}
\put(45,-15){\line(0,1){57}}\put(45,42){\vector(1,0){15}}
\put(-4,-25){${\bf x}^1$}\put(11,-25){${\bf x}^2$}
\put(25,-25){$\cdots$}\put(41,-25){${\bf x}^n$}
\put(0,-2){\put(18.4,44){$\pi^n({\bf x})$}
\multiput(12,25)(3,3){5}{\put(0,0){.}}
\put(-11,14){$\pi^2({\bf x})$}\put(-26,-1){$\pi^1({\bf x})$}}
\multiput(53,20)(0,3.3){5}{\put(0,0){.}}
\end{picture}
\end{equation}
We have bent the arrows so that there are $n$ incoming ones 
from the bottom and $n$ outgoing ones to the right.
This turns out to be a natural shape when the diagram is embedded into a 
layer-to-layer transfer matrix of 3D lattice model 
associated with the tetrahedron equation \cite{KMO2,KMO4}.
Now we are ready to show 
that $\pi^1({\bf x}) \otimes \cdots \otimes \pi^n({\bf x}) 
\in B_+({\bf m})$ indeed holds in (\ref{sae}).
\begin{lemma}\label{le:sae}
$\pi^a({\bf x})$ defined by (\ref{srn}) satisfies
$\pi^1({\bf x}) \ge \cdots \ge \pi^n({\bf x})$.
\end{lemma}
\begin{proof}
The diagram (\ref{kan:t}) tells that 
$R_{\ell_a, \ell_{a+1}}(\pi^a({\bf x}) \otimes {\bf u} ) = \pi^{a+1}({\bf x}) \otimes {\bf v}$
holds for some ${\bf u} \in B_{\ell_{a+1}}$ and 
${\bf v} \in B_{\ell_a}$.
Then the assertion follows from (\ref{mri}).
\end{proof}

The diagram (\ref{kan:t}) reminds us of a corner transfer matrix \cite[Chap.13]{Bax}.
In fact,  in the forthcoming Theorem \ref{th:mnm} 
it will be used exactly as the matrix element of it 
for the $U_q(\widehat{sl}_L)$-vertex model associated with the 
symmetric tensor representations of degrees $\ell_1, \ell_2, \ldots, \ell_n$
at $q=0$, where every vertex is frozen to the combinatorial $R$ and 
all the edge variables are uniquely determined from the input
${\bf x}^1, {\bf x}^2, \ldots, {\bf x}^n$.

\begin{example}\label{ex:mst}
Take again
${\bf x}=114\otimes 202 \otimes 210 \otimes 001$ in Example \ref{ex:tsk}.
Then $\pi^1({\bf x}), \ldots, \pi^4({\bf x})$ are calculated as
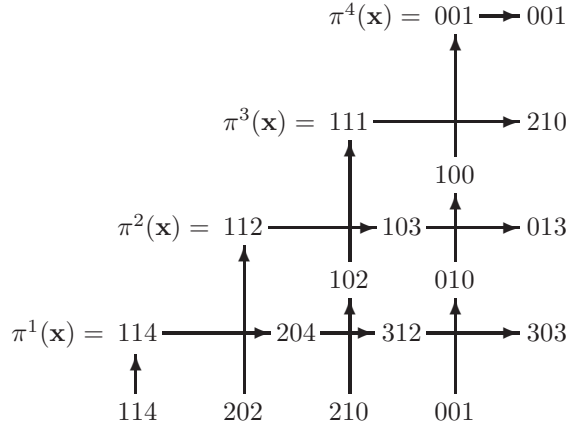
\begin{equation*}
\begin{picture}(140,160)(0,-20)
\thicklines

\put(70,130){$\pi^4({\bf x})=$}
\put(110,130){001}
\put(127,133){\vector(1,0){15}}\put(145,130){001}

\put(145,90){210}

\put(30,90){$\pi^3({\bf x})=$}
\put(70,90){111}
\put(87,93){\vector(1,0){55}}
\put(110,70){100}\put(118,80){\vector(0,1){46}}

\put(90,50){103}\put(107,53){\vector(1,0){35}}\put(145,50){013}

\put(-10,50){$\pi^2({\bf x})=$}
\put(30,50){112}
\put(47,53){\vector(1,0){41}}
\put(70,30){102}\put(78,40){\vector(0,1){46}}
\put(110,30){010}\put(118,40){\vector(0,1){26}}

\put(50,10){204} \put(67,13){\vector(1,0){20}}
\put(90,10){312}\put(145,10){303}\put(107,13){\vector(1,0){35}}

\put(-50,10){$\pi^1({\bf x})=$}
\put(-10,10){114}
\put(-10,-20){114}\put(-3,-10){\vector(0,1){15}}\put(7,13){\vector(1,0){41}}
\put(30,-20){202}\put(38,-10){\vector(0,1){56}}
\put(70,-20){210}\put(78,-10){\vector(0,1){36}}
\put(110,-20){001}\put(118,-10){\vector(0,1){36}}

\end{picture}
\end{equation*}
They agree with those obtained in Example \ref{ex:srn}. 
\end{example}

\subsection{Queueing type algorithm for $\pi$ and TAZRP embedding}\label{ss:mst}
Here we explain more human-friendly perceivable algorithm
to derive the $n$-TAZRP state $\pi({\bf x})\in S({\bf m})$ in multiset representation
from an $n$-line state 
${\bf x} = {\bf x}^1\otimes \cdots \otimes {\bf x}^n \in B({\bf m})$.
We illustrate it along the same example 
${\bf x}=114\otimes 202 \otimes 210 \otimes 001$ and 
$\pi({\bf x})= (3,3,1124)$ with $(n,L)=(4,3)$ as 
Example \ref{ex:tsk} and \ref{ex:srn}.

First draw the dot diagram (I) of ${\bf x} = {\bf x}^1\otimes \dots \otimes {\bf x}^n$
as in (\ref{hrk}).
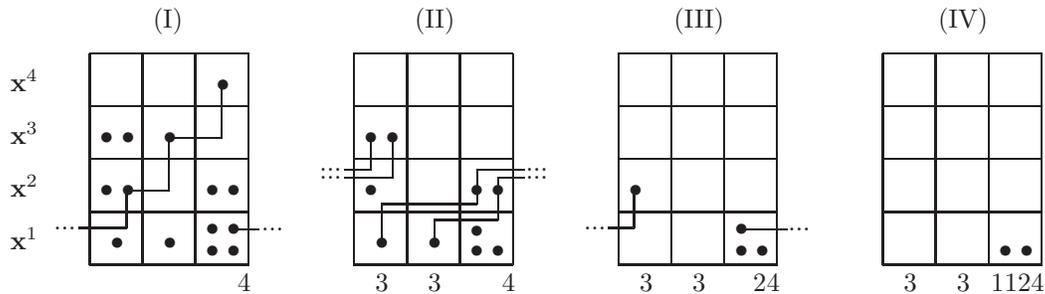
\begin{equation*}
\begin{picture}(400,110)(-35,-10)
%%%%%%4

\put(24,90){(I)}
\multiput(0,0)(20,0){4}{\put(0,0){\line(0,1){80}}}
\multiput(0,0)(0,20){5}{\put(0,0){\line(1,0){60}}}

\put(-30,66){${\bf x}^4$}
\put(-30,46){${\bf x}^3$}
\put(-30,26){${\bf x}^2$}
\put(-30,6){${\bf x}^1$}

\put(48,66){$\bullet$}
\put(4,46){$\bullet$}\put(12,46){$\bullet$}
\put(28,46){$\bullet$}
\put(4,26){$\bullet$}\put(12,26){$\bullet$}
\put(44,26){$\bullet$}\put(52,26){$\bullet$}
\put(8,6){$\bullet$}\put(28,6){$\bullet$}
\put(44,11){$\bullet$}\put(52,11){$\bullet$}
\put(44,3){$\bullet$}\put(52,3){$\bullet$}
%line
\put(50,68){\line(0,-1){20}}\put(50,48){\line(-1,0){20}}
\put(30,48){\line(0,-1){20}}\put(30,28){\line(-1,0){14}}
\put(14,28){\line(0,-1){14}}\put(14,14){\line(-1,0){18}}
\put(-13.5,13.5){...}
\put(53,13.5){\line(1,0){11}}\put(65,13){...}
\put(41,-10){$\phantom{112}4$}

%%%%%%3
\put(100,0){
\put(23,90){(II)}
\multiput(0,0)(20,0){4}{\put(0,0){\line(0,1){80}}}
\multiput(0,0)(0,20){5}{\put(0,0){\line(1,0){60}}}
\put(4,46){$\bullet$}\put(12,46){$\bullet$}
\put(4,26){$\bullet$}
\put(44,26){$\bullet$}\put(52,26){$\bullet$}
\put(8,6){$\bullet$}\put(28,6){$\bullet$}
\put(44,10.5){$\bullet$}
\put(44,3){$\bullet$}\put(52,3){$\bullet$}

\put(6.5,47){\line(0,-1){11}}\put(6.5,36){\line(-1,0){10}}\put(-13,35.5){...}
\put(14.5,47){\line(0,-1){14}}\put(14.5,33){\line(-1,0){18}}\put(-13,32.5){...}
\put(46.5,27){\line(0,1){9}}\put(46.5,36){\line(1,0){18}}\put(65,35.5){...}
\put(54.5,27){\line(0,1){6}}\put(54.5,33){\line(1,0){10}}\put(65,32.5){...}
\put(10.5,7){\line(0,1){16}}\put(10.5,23){\line(1,0){36}}
\put(46.5,23){\line(0,1){4}}
\put(54.5,27){\line(0,-1){10}}\put(54.5,17){\line(-1,0){24}}
\put(30.5,17){\line(0,-1){8}}

\put(8,-10){3}\put(28,-10){3}
\put(41,-10){$\phantom{112}4$}
}

%%%%%%2
\put(200,0){
\put(21,90){(III)}
\multiput(0,0)(20,0){4}{\put(0,0){\line(0,1){80}}}
\multiput(0,0)(0,20){5}{\put(0,0){\line(1,0){60}}}
\put(4,26){$\bullet$}
\put(44,11){$\bullet$}
\put(44,3){$\bullet$}\put(52,3){$\bullet$}
\put(6,28){\line(0,-1){14}}\put(6,14){\line(-1,0){10}}\put(-13,13.6){...}
\put(45,13.6){\line(1,0){19}}\put(64.3,13.2){...}

\put(8,-10){3}\put(28,-10){3}
\put(41,-10){$\phantom{11}24$}
}

%%%%%%1
\put(300,0){
\put(21,90){(IV)}
\multiput(0,0)(20,0){4}{\put(0,0){\line(0,1){80}}}
\multiput(0,0)(0,20){5}{\put(0,0){\line(1,0){60}}}
\put(44,3){$\bullet$}\put(52,3){$\bullet$}
\put(8,-10){3}\put(28,-10){3}
\put(41,-10){$1124$}
}

\end{picture}
\end{equation*}
Do the following procedure for $a=n, n-1, \ldots, 1$ in this order.

\begin{quote}
Draw an $H$-line from each dot in ${\bf x}^a$ 
by applying the NY-rule repeatedly until it reaches 
some dot in the bottom row which belongs to ${\bf x}^1$.
Record the captured dots in  ${\bf x}^1$ as particle of species $a$ and 
erase all the dots connected by the $H$-lines.
\end{quote}
The array of multiset of particles gives the $n$-TAZRP configuration,
i.e., the image of $\pi$. 

For the procedure with $a=1$, no $H$-line needs to be drawn and one just assigns $1$ 
to the existing dots.
For each $a$,
the $H$-lines depend on the order of picking the dots in ${\bf x}^a$
but the final output of the procedure does not depend on it 
thanks to Remark \ref{re:mri} (1).
In the present example, 
one gets (I)$\rightarrow$ (II)$\rightarrow$ (III)$\rightarrow$ (IV)
as the procedure is executed for $a=4,3,2,1$.
(We omitted the empty diagram produced by the last $a=1$ case.)

The equivalence of the above algorithm for $\pi$ and the definition (\ref{sae}) 
is shown easily if one notices that the composition of 
the combinatorial $R$ to get $\pi^a({\bf x})$ as in Example \ref{ex:srn}
is described as a repeated application of the NY-rule in a dot diagram.

One can further reformulate the algorithm inductively with respect to $n$ so as to produce 
$n$-TAZRP states from $(n-1)$-TAZRP states and $B_{\ell_1}$.
Consider the above example for $n=4$.
The ${\bf x}^2 \otimes{\bf x}^3 \otimes{\bf x}^4$ 
without the bottom row ${\bf x}^1$ is an $3$-line state 
whose projection is the $3$-TAZRP state $(13,\emptyset,22)$.
Increase the species uniformly by 1 to get $(24,\emptyset,33)$\footnote{
This minor extra process can be avoided 
if the species $a$ is replaced by $n+1-a$ everywhere.} and (V) below.
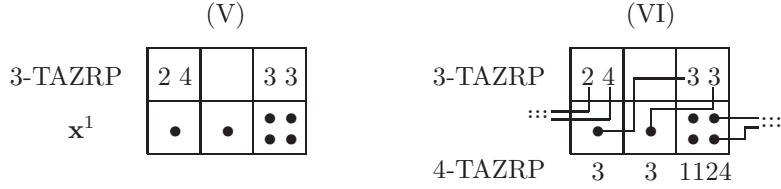
\begin{equation*}
\begin{picture}(300,70)(-60,-10)
%%%%%%4

\put(22,50){(V)}
\multiput(0,0)(20,0){4}{\put(0,0){\line(0,1){40}}}
\multiput(0,0)(0,20){3}{\put(0,0){\line(1,0){60}}}

\put(-52,26){3-TAZRP}
\put(-30,6){${\bf x}^1$}

\put(4,26){$2$}\put(12,26){$4$}
\put(44,26){$3$}\put(52,26){$3$}
\put(8,6){$\bullet$}\put(28,6){$\bullet$}
\put(44,11){$\bullet$}\put(52,11){$\bullet$}
\put(44,3){$\bullet$}\put(52,3){$\bullet$}
%line

\put(160,0){
\put(21,50){(VI)}
\multiput(0,0)(20,0){4}{\put(0,0){\line(0,1){40}}}
\multiput(0,0)(0,20){3}{\put(0,0){\line(1,0){60}}}

\put(-52,26){3-TAZRP}

\put(-52,-9){4-TAZRP}

\put(4,26){$2$}\put(12,26){$4$}
\put(44,26){$3$}\put(52,26){$3$}
\put(8,6){$\bullet$}\put(28,6){$\bullet$}
\put(44,11){$\bullet$}\put(52,11){$\bullet$}
\put(44,3){$\bullet$}\put(52,3){$\bullet$}
%line

%from 2
\put(-7,-3){\put(14,28){\line(0,-1){9}}\put(14,19){\line(-1,0){14}}
\put(-9.5,18.5){...}}

%from 4
\put(1,-3){\put(14,28){\line(0,-1){12}}\put(14,16){\line(-1,0){22}}
\put(-17.5,16.5){...}}

%For left 3
\put(10,8.5){\line(1,0){14}}\put(24,8.5){\line(0,1){20}}\put(24,28.5){\line(1,0){19}}

% from right 3
\put(54,25){\line(0,-1){8}}\put(54,17){\line(-1,0){23.5}}\put(30.5,17){\line(0,-1){10}}

\put(53,13.5){\line(1,0){18}}\put(72,13){...}
\put(53,5.5){\line(1,0){13}}\put(66,5.5){\line(0,1){4.5}}
\put(66,10){\line(1,0){5}}
\put(72,10){...}

\put(8,-10){3}\put(28,-10){3}
\put(41,-10){$1124$}
}

\end{picture}
\end{equation*}
Regard the 3-TAZRP state as a collection of particles with species 
$a=2,3, \ldots, n$ ($n=4$ in our ongoing illustration).
Draw $H$-lines from them to the dots in ${\bf x}^1$ 
{\em in the order $a=n,n-1,\ldots, 2$}  according to the NY-rule.
For each $a$, the set of dots linked with the particles of species $a$ 
is independent of the order of picking the particles due to Remark \ref{re:mri} (1).
Regard the dots in ${\bf x}^1\in B_{\ell_1}$ connected to the particles $a$ 
also as particles $a$. 
Dots in ${\bf x}^1\in B_{\ell_1}$ not captured by any $H$-line is regard as particle $1$. 
Then the bottom row gives the $n$-TAZRP state.
See (VI).
In general the procedure (V) $\rightarrow$ (VI) to get $n$-TAZRP states from 
$B_{\ell_1}$ and $(n-1)$-TAZRP states (with species from $[2,n]$)
is a modification of the NY-rule in that the one set of the dots is 
assigned with species and larger species ones have the 
priority to emanate an $H$-line.
We call it the {\em TAZRP embedding rule}.
For $n$-line states ${\bf x}={\bf x}^1\otimes \cdots \otimes {\bf x}^n
\in B({\bf m})$, set
\begin{align*}
\{(k-1)-\text{TAZRP states}\}\times B_{\ell_{n-k+1}}\;
&\longrightarrow \{k-\text{TAZRP states}\}\\
({\boldsymbol \sigma}, {\bf x}^{n-k+1}) \quad
&\longmapsto \;\;
\Phi_{k,{\bf x}^{n-k+1}}({\boldsymbol \sigma})\qquad
(k \in [2,n]).
\end{align*}
Here $\Phi_{k,{\bf x}^{n-k+1}}$ signifies the TAZRP embedding rule;
one increases the species of particles in ${\boldsymbol \sigma}$ by 1 and 
uses the resulting state as the top row and ${\bf x}^{n-k+1}$ as the bottom row
in the diagram like (VI) to produce a $k$-TAZRP state.
Regarding ${\bf x}^n$ as a $1$-TAZRP state naturally, we have
\begin{align*}
\pi({\bf x}^1\otimes \cdots \otimes {\bf x}^n)
= \Phi_{n,{\bf x}^1}\circ  \Phi_{n-1,{\bf x}^2}\circ
\cdots \circ  \Phi_{2,{\bf x}^{n-1}}({\bf x}^n).
\end{align*}
This construction is a TAZRP analogue of the nested Bethe ansatz
which diagonalizes the Hamiltonian of integrable $sl(n)$ spin chains 
inductively on $n$. 

The algorithms explained in this subsection are the TAZRP counterpart of the multiple 
queueing process introduced for the TASEP \cite{FM}.
Its reformulation by crystals and combinatorial $R$ in \cite{KMO} 
is quite parallel with the content here.
They offer a unified perspective into the 
multispecies TAZRP and TASEP on the periodic chain $\Z_L$ 
as the sister models corresponding to the  
symmetric and the antisymmetric tensor representations of 
the quantum affine algebra $U_q(\widehat{sl}_L)$ at $q=0$.

\subsection{Induced dynamics}

We extend the map $\pi$ (\ref{sae}) naturally 
to that on the space of states for $n$-line process and $n$-TAZRP
$\pi: \bigoplus_{{\bf x} \in B({\bf m})} \C|{\bf x}\rangle
\rightarrow \bigoplus_{{\boldsymbol \sigma} \in S({\bf m})}
\C|{\boldsymbol \sigma}\rangle$ by linearity and 
$\pi(|{\bf x}\rangle) = |\pi({\bf x})\rangle$.
By the construction $\pi$ is surjective.
Recall that $\tau^k_i$ is defined around (\ref{kgc2}) 
and $T^k_{a,i}$ by (\ref{skb2}) and (\ref{skb:f}).

\begin{proposition}\label{pr:sir}
For $(i,a,k) \in \Z_L \times [1,n] \times \Z_{\ge 1}$, 
set $\tilde{\tau}^k_{i,a} = \tau^k_i$ for $a=1$ and 
$\tilde{\tau}^k_{i,a} = 1$ for $a \in [2,n]$.
Then the following diagram is commutative:
\begin{equation*}
\begin{CD}
B({\bf m}) @>{T^k_{i,a}}>> B({\bf m})\\
@V{\pi}VV@VV{\pi}V\\
S({\bf m}) @>>{\tilde{\tau}^k_{i,a}}> S({\bf m})
\end{CD}
\end{equation*}
\end{proposition}

\begin{proof}
We invoke the induction on $n$.
For $n=1$ the claim is obvious.
Assume the claim for $(n-1)$-TAZRP.
We utilize the description of $\pi$ by 
the TAZRP embedding rule explained in the end of 
Section \ref{ss:mst}. 
Suppose $T^k_{i,a}: {\bf x}^1\otimes {\bf x}^{\ge 2} 
\mapsto {\bf y}^1\otimes {\bf y}^{\ge 2} $ with
${\bf x}^{\ge 2}, {\bf y}^{\ge 2} \in 
B_{\ell_2}\otimes \cdots \otimes B_{\ell_n}$.
Then $T^{k'}_{i,a'}({\bf x}^{\ge 2}) = {\bf y}^{\ge 2}$ holds
for $(n-1)$-TAZRP states for some $a'$ and $k'$.
Set $\pi({\bf x}^{\ge 2}) +1 = 
(\ldots, \alpha_1\ldots \alpha_r,\beta_1\ldots \beta_s, \ldots)$,
where ``$+1$" stands for the uniform increment of the species by 1, and 
$\alpha_1\ldots \alpha_r$ and $\beta_1\ldots \beta_s$
are the resulting $(n-1)$-TAZRP local states (with species from $[2,n]$) 
at the $i$-th and the $(i+1)$-th site, respectively.
In particular $2 \le \beta_1 \le \cdots \le \beta_s \le n$.
Set ${\bf x}^1=(\ldots, u,v,\ldots)$ similarly exhibiting 
the $i$-th and the $(i+1)$-th components. 
By the induction assumption we know 
$\pi({\bf y}^{\ge 2}) +1 = 
(\ldots, \alpha_1\ldots \alpha_r\beta_1\ldots \beta_t, 
\beta_{t+1}\ldots \beta_s,\ldots)$ for some $t \in [0,s]$.
From the definition of $T^k_{i,a}$ it follows that 
${\bf y}^1=(\ldots, u+w,v-w,\ldots)$ for some $w \in [0,v]$.
Now the part of the diagram (V) in Section \ref{ss:mst}
corresponding to the $i$-th and the $(i+1)$-th columns looks as follows.
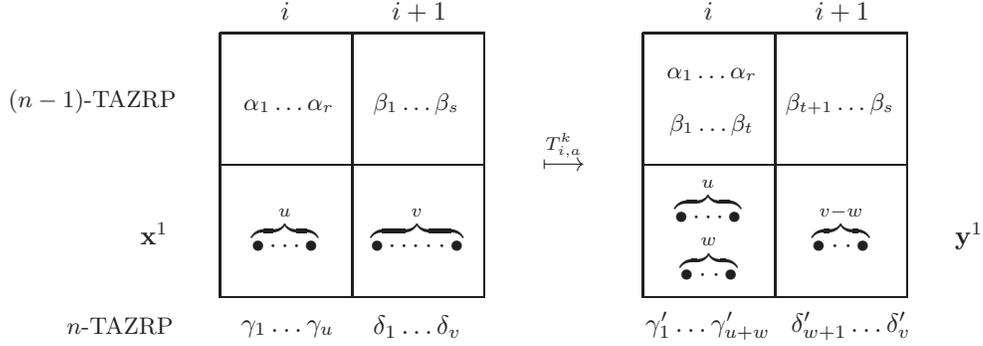
\begin{equation*}
\begin{picture}(300,140)(-40,-20)
%%%%%%rhs

\put(-80,72){\small $(n-1)$-TAZRP}
\put(-30,20){${\bf x}^1$}\put(278,20){${\bf y}^1$}
\put(-58,-14){\small $n$-TAZRP}

\put(23,105){$i$}\put(65,105){$i+1$}
\multiput(0,0)(50,0){3}{\put(0,0){\line(0,1){100}}}
\multiput(0,0)(0,50){3}{\put(0,0){\line(1,0){100}}}

\put(9,71){\small $\alpha_1\ldots\alpha_r$}

\put(58,71){\small $\beta_{1} \ldots \beta_s$}

\put(12,17){$\overbrace{\bullet \cdots \bullet}^u$}

\put(57,17){$\overbrace{\bullet \cdots \cdot \cdot \,\bullet}^v$}

\put(8,-13){$\gamma_1\ldots \gamma_u$}
\put(58,-14){$\delta_1\ldots \delta_v$}

\put(122,48){$\overset{T^k_{i,a}}{\longmapsto}$}

\put(160,0){
\put(23,105){$i$}\put(65,105){$i+1$}
\multiput(0,0)(50,0){3}{\put(0,0){\line(0,1){100}}}
\multiput(0,0)(0,50){3}{\put(0,0){\line(1,0){100}}}

\put(9,83){\small $\alpha_1\ldots\alpha_r$}
\put(10,63){\small $\beta_1 \ldots \beta_t$}

\put(54,71){\small $\beta_{t+1} \ldots \beta_s$}

\put(12,28){$\overbrace{\bullet \cdots \bullet}^u$}
\put(14,6){$\overbrace{\bullet \cdot \cdot \,\bullet}^w$}

\put(64,17){$\overbrace{\bullet \cdot \cdot \,\bullet}^{v-w}$}

\put(1,-13){$\gamma'_1\ldots \gamma'_{u+w}$}
\put(56,-13){$\delta'_{w+1}\ldots \delta'_{v}$}

}
\end{picture}
\end{equation*}
We assume 
$2 \le \delta_1 \le \cdots \le \delta_v \le n$.
We are to draw $H$-lines from the particles in the top row
to some dot in the bottom 
according to the TAZRP embedding rule in Section \ref{ss:mst}.
To clarify the argument below, 
we assume that the $H$-line from a particle in the top row first goes down vertically,
make $90^\circ$ right turn in the box underneath and proceeds horizontally
to the left periodically until it captures a yet unconnected dot in the bottom row.
Thus in the bottom row of the above diagram,
there are $H$-lines coming from the right of the $(i+1)$-th box 
seeking the target dots and also the outgoing ones to the left of the $i$-th box. 
The new $n$-TAZRP particles $\gamma'_j$ and $\delta'_j$ can be 
related to the previous ones $\gamma_j$ and $\delta_j$  
by inspecting the influence of the changes of the diagram on the $H$-lines.
There are three cases to consider.

Case 1. $t\ge 1$ and $w=0$.
From (\ref{syr}) this happens only if $a=2$ and $u\le s-t$ with $k=t$.
In the left diagram,  the $H$-lines from $\beta_{s-u+1},\ldots, \beta_s$
are captured by the $u$ dots in the bottom left box.
The other $H$-lines from $\beta_1,\ldots, \beta_{s-u}$ are outgoing to the left of it. 
This situation is the same as the right diagram, showing that the relocation of 
the particles $\beta_1, \ldots ,\beta_t$ causes no change in the 
result of the TAZRP embedding rule. 
Thus we find $\gamma'_j = \gamma_j$ and $\delta'_j = \delta_j$
in agreement with $\tilde{\tau}^k_{i,2}=1$.

Case 2. $t=0$ and $w\ge 1$.
This happens only if $a=1$ and $u\ge s$ with $k=w$.
In the both diagram, the $s$ $H$-lines from $\beta_1,\ldots, \beta_s$ are 
all captured by the dots in the bottom left box.
Thus the $H$-lines outgoing to the left of the $i$-th box is the same.
It implies that the $H$-lines coming from the right of 
the $(i+1)$-th box are also unchanged. 
Since the particles with larger species have priority to find the partner dots
in the TAZRP embedding rule, 
we have $\delta'_j = \delta_j$ for $j \in [w+1,v]$ and 
$\{\gamma'_j\} = \{\gamma_j\}\cup \{\delta_1,\ldots, \delta_w\}$.
This agrees with (\ref{kgc2}) for $\tau^k_i = \tilde{\tau}^k_{i,a=1}$.

Case 3. $t\ge 1$ and $w \ge 1$.
This happens only if $a=1$ and $s-t=u$ with $k=w$. 
In the both diagram, the $u$ $H$-lines from $\beta_{t+1},\ldots, \beta_s$
are absorbed into the dots in the bottom left box, and the other ones 
from $\beta_1,\ldots, \beta_t$ are outgoing to the left of the $i$-th box.
Then the rest of the argument is the same as Case 2.
\end{proof}

\begin{example}\label{ex:mst2}
Consider the maps in Example \ref{ex:mho}.
According to the definition (\ref{skb2}),
we have $T^1_{1,1}({\bf x}), \ldots, T^1_{3,1}({\bf x}) \in B(2,1,2,1)$ for 
${\bf x}$ given in Example \ref{ex:tsk}.
Their image by $\pi$ are given by
\begin{align*}
\pi(T^1_{1,1}({\bf x})) &= |33,\emptyset,1124\rangle,\quad
\pi(T^1_{2,1}({\bf x}))= |3,13,124\rangle,\quad \;\;
\pi(T^2_{2,1}({\bf x}))= |3,113,24\rangle,\\
\pi(T^3_{2,1}({\bf x}))&= |3,1123,4\rangle,\quad\;\;
\pi(T^4_{2,1}({\bf x}))= |3,11234,\emptyset\rangle,\quad
\pi(T^1_{2,2}({\bf x}))= |3,3,1124\rangle,\\
\pi(T^1_{3,1}({\bf x}))&= |\emptyset,3,11234\rangle.
\end{align*}
By the result of Example \ref{ex:srn} on the same ${\bf x}$,
the claim of Proposition \ref{pr:sir} is rephrased as
\begin{align*}
\tau^1_1|3,3,1124\rangle &=  |33,\emptyset,1124\rangle,\quad
\tau^1_2|3,3,1124\rangle = |3,13,124\rangle,\quad
\tau^2_2|3,3,1124\rangle = |3,113,24\rangle,\\
\tau^3_2|3,3,1124\rangle &=  |3,1123,4\rangle,\quad\;\;
\tau^4_2|3,3,1124\rangle= |3,11234,\emptyset\rangle,\quad
1 |3,3,1124\rangle = |3,3,1124\rangle,\\
\tau^1_3|3,3,1124\rangle &= |\emptyset,3,11234\rangle.
\end{align*}
These relations agree with the definition (\ref{kgc2}).
\end{example}

Proposition \ref{pr:sir} tells that the dynamics of $n$-TAZRP is 
exactly the one that is induced from $n$-line process via the map $\pi$.
Now we state the first main result of the article.

\begin{theorem}[Steady state of $n$-TAZRP]\label{th:szk}
The steady state $|\bar{P}_L({\bf m})\rangle$ 
of $n$-$\mathrm{TAZRP}$ in the sector $S({\bf m})$ is the 
image of the $n$-line process steady state $|\Omega({\bf m})\rangle$ 
in Corollary \ref{akn:u} by $\pi$. 
Namely,
\begin{align*}
|\bar{P}_L({\bf m})\rangle = 
\pi |\Omega({\bf m})\rangle = \sum_{{\bf x} \in B({\bf m})}
|\pi({\bf x})\rangle.
\end{align*}
\end{theorem}

\begin{proof}
By Proposition \ref{pr:sir}, the Markov matrix of $n$-TAZRP (\ref{ymi:s})
is expressed as\\
$H_{\mathrm{TAZRP}} =\sum_{i\in \Z_L} \sum_{a=1}^n\sum_{k \ge 1}
(\tilde{\tau}^k_{i,a}-1)$.
Moreover the intertwining relation 
\begin{align*}
\pi H_{\mathrm{LP}} = H_{\mathrm{TAZRP}}\,  \pi
\end{align*}
holds,
where $H_{\mathrm{LP}}$ is the Markov matrix of $n$-line process 
(\ref{skb:t}).
Thus we have $H_{\mathrm{TAZRP}} |\bar{P}_L({\bf m})\rangle = 0$ from Corollary \ref{akn:u}.
This proves the claim thanks to the uniqueness of the steady state.
\end{proof}
Theorem \ref{th:szk}  is an analogue of the combinatorial construction of the 
$n$-TASEP steady state \cite{FM}
whose quantum group theoretical origin was uncovered in \cite{KMO}.

Before closing the section 
we note a simple consequence of Theorem \ref{th:szk}.
Consider the most localized configuration of the $n$-TAZRP 
$(\emptyset, \ldots, \emptyset, all)$ and its cyclic permutations
which have the same probability.
Here $all$ means the assembly of all the $\ell_1$ particles in 
the sector $S({\bf m})$. 
See (\ref{mkr:akci}) for $\ell_a$.
It is easy to see 
\begin{align*}
\pi^{-1}(\emptyset, \ldots, \emptyset, all)= (0,\ldots,0,\ell_1)
\otimes B_{\ell_2} \otimes \cdots \otimes B_{\ell_n}.
\end{align*}
Therefore we get 
\begin{align}\label{kyk:sn}
{\mathbb P}(\emptyset, \ldots, \emptyset, all) 
 = \prod_{a=2}^n \binom{L-1+\ell_a}{\ell_a}
 \end{align}
in agreement with Example \ref{ex:LL}.
This is the {\em largest} probability in the sector.
Similarly we find
\begin{align*}
{\mathbb P}(\ast,\overbrace{\emptyset,\ldots, \emptyset}^d,
\sigma_{d+1},\ldots, \sigma_L)=
\binom{d+\ell_2}{\ell_2}\prod_{a=3}^n \binom{L-1+\ell_a}{\ell_a}
\end{align*}
if $\sigma_{d+1} \neq \emptyset$ and 
$\sigma_{d+1},\ldots, \sigma_L$ contain particle of species $1$ only with 
$\ast$ being the rest.

\section{Formulae for steady state probability}\label{sec:f}
\subsection{Crystalline corner transfer matrix} 
Recall that the steady state of $n$-TAZRP on $\Z_L$ in the sector $S({\bf m})$ 
has the form (\ref{ymi:cu}).
Theorem \ref{th:szk} tells that the steady state probability therein is expressed as
\begin{align}\label{szk:i}
\mathbb{P}({\boldsymbol \sigma} )
= \#\{{\bf x} \in B({\bf m})\mid \pi({\bf x}) = {\boldsymbol \sigma} \}.
\end{align}

\begin{example}\label{ex:szk}
Consider $3$-TAZRP on $\Z_3$ in the sector $S(1,2,1)$.
Example \ref{ex:LL} tells ${\mathbb P}(1,2,23) = 5$,
where $(1,2,23)$ is multiset representation.
In fact there are exactly 5 elements in $B(1,2,1)$ which are mapped 
by $\pi$ to $(1,2,23)$.
In the notation (\ref{hrk}) they are given by
\begin{align*}
\begin{pmatrix}
 1 & 0 & 0 \\
 3 & 0 & 0 \\
 1 & 1 & 2
\end{pmatrix},
\quad
\begin{pmatrix}
 0 & 1 & 0 \\
 3 & 0 & 0 \\
 1 & 1 & 2
\end{pmatrix},
\quad
\begin{pmatrix}
 0 & 0 & 1 \\
 3 & 0 & 0 \\
 1 & 1 & 2
\end{pmatrix},
\quad
\begin{pmatrix}
 0 & 1 & 0 \\
 2 & 0 & 1 \\
 1 & 1 & 2
\end{pmatrix},
\quad
\begin{pmatrix}
 0 & 0 & 1 \\
 2 & 0 & 1 \\
 1 & 1 & 2
\end{pmatrix}.
\end{align*}
\end{example}

From (\ref{sra}) we see 
$\varphi \circ \pi=\pi^1\otimes \cdots \otimes \pi^n$.
Thus the condition  
$\pi({\bf x}) = {\boldsymbol \sigma}$ in (\ref{szk:i}) 
is equivalent to 
$\pi^1({\bf x}) \otimes \cdots \otimes \pi^n({\bf x})  = 
\varphi({\boldsymbol \sigma}) = 
\varphi^1({\boldsymbol \sigma}) \otimes \cdots \otimes
\varphi^n({\boldsymbol \sigma})$.
Therefore (\ref{szk:i}) is rewritten as
\begin{align}\label{szk:t}
\mathbb{P}({\boldsymbol \sigma} )
= \#\{{\bf x} \in B({\bf m})\mid \pi^a({\bf x}) = \varphi^a({\boldsymbol \sigma}), 
\forall a \in [1,n] \}.
\end{align}

To proceed we find it convenient to 
generalize the vertex diagram (\ref{ask3}) for 
$R=R_{\ell,m}$ naturally to arbitrary edge states 
${\bf a} = (a_1,\ldots, a_L),
{\bf i} = (i_1,\ldots, i_L) \in B_\ell$ and 
${\bf b} = (b_1,\ldots, b_L),
{\bf j} = (j_1,\ldots, j_L) \in B_m$ as
\begin{equation}\label{lin:ok}
\begin{picture}(50,40)(50,-18)
\thicklines

\put(-18,-3){${\bf i}$}
\put(14,-2){${\bf a}$}
\put(-2,-19){${\bf j}$}
\put(-3.5,14){${\bf b}$}
\put(-10,0){\vector(1,0){20}}
\put(0,-10){\vector(0,1){20}}

\put(24,-2){$\displaystyle{= R^{{\bf a}, {\bf b}}_{{\bf i}, {\bf j}} = 
\begin{cases}
1 & \text{if}\; R({\bf i} \otimes {\bf j}) = {\bf b} \otimes {\bf a},\\
0 & \text{otherwise}.
\end{cases}}$}
\end{picture}
\end{equation}
By the definition $ R^{{\bf a}, {\bf b}}_{{\bf i}, {\bf j}}=0$ unless
${\bf a}+{\bf b}={\bf i}+{\bf j}$. 
This property generalizing the ice rule in the six-vertex model \cite{Bax} 
will be referred to as the {\em weight conservation}.
The matrix element $R^{{\bf a}, {\bf b}}_{{\bf i}, {\bf j}}$ is 
nothing but the Boltzmann weight of the local vertex configuration at $q=0$. 
Concatenations of the vertices in diagrams 
are naturally interpreted as configuration sums.
With this convention we have
\begin{theorem}[Corner transfer matrix formula]\label{th:mnm}
The steady state probability ${\mathbb P}({\boldsymbol \sigma})$ 
of $n$-$\mathrm{TAZRP}$ in the sector $S({\bf m})$ is expressed as 
\begin{equation*}
\begin{picture}(150,75)(-100,-23)
\thicklines

\put(-143,15){\large ${\mathbb P}({\boldsymbol \sigma})
=\displaystyle{\sum_{{\bf x}^1 \otimes \cdots \otimes {\bf x}^n \in B({\bf m})}}$}
\put(0,-15){\line(0,1){12}}\put(0,-3){\vector(1,0){60}}
\put(15,-15){\line(0,1){27}}\put(15,12){\vector(1,0){45}}
\put(45,-15){\line(0,1){57}}\put(45,42){\vector(1,0){15}}
\put(-4,-25){${\bf x}^1$}\put(11,-25){${\bf x}^2$}
\put(25,-25){$\cdots$}\put(41,-25){${\bf x}^n$}
\put(0,-2){\put(17,44){$\varphi^n({\boldsymbol \sigma})$}
\multiput(12,25)(3,3){5}{\put(0,0){.}}
\put(-12.3,14){$\varphi^2({\boldsymbol \sigma})$}\put(-27.3,-1){$\varphi^1({\boldsymbol \sigma})$}}
\multiput(53,20)(0,3.3){5}{\put(0,0){.}}
\end{picture}
\end{equation*}
The right hand side stands for the configuration sum 
under the boundary condition specified by 
$\varphi^1({\boldsymbol \sigma}), \ldots, \varphi^n({\boldsymbol \sigma})$.
\end{theorem}
\begin{proof}
Follows from (\ref{szk:t})  and (\ref{kan:t}).
\end{proof}

\begin{example}\label{ex:lin}
Consider the $3$-TAZRP on $\Z_3$ in the sector $S(1,2,1)$,
which is the same as Example \ref{ex:szk}.
We have $\varphi(1,2,23)= 112\otimes 012 \otimes 001$
according to the rule illustrated in Example \ref{ex:ask:g}.
By using them as the boundary condition,
${\mathbb P}(1,2,23) = 5$ is derived as the following sum 
corresponding to the 5 elements 
${\bf x}^1\otimes{\bf x}^2\otimes{\bf x}^3\in B(1,2,1)$ 
in Example \ref{ex:szk}.
\begin{equation*}
\begin{picture}(150,70)(120,-30)
\thicklines

\put(0,-15){\line(0,1){12}}
\put(0,-3){\line(1,0){18.7}}\put(27,-3){\vector(1,0){18}}
\put(15,-15){\line(0,1){27}}\put(15,12){\vector(1,0){30}}
\put(30,-15){\line(0,1){15.5}} \put(30,27){\line(0,-1){18.5}}
\put(30,27){\vector(1,0){15}}
\put(19.5,-5){${\bf u}$}\put(26,2.2){${\bf w}$}
\put(7,15){012}
\put(48,24){$001$}\put(48,9){$111$}\put(48,-6){$400$}
\put(-11,-25){112}\put(8.5,-25){300}\put(27,-25){100}
\put(71,6){$+$}

\put(85,0){
\put(0,-15){\line(0,1){12}}
\put(0,-3){\line(1,0){18.7}}\put(27,-3){\vector(1,0){18}}
\put(15,-15){\line(0,1){27}}\put(15,12){\vector(1,0){30}}
\put(30,-15){\line(0,1){15.5}} \put(30,27){\line(0,-1){18.5}}
\put(30,27){\vector(1,0){15}}
\put(19.5,-5){${\bf u}$}\put(26,2.2){${\bf w}$}
\put(7,15){012}
\put(48,24){$001$}\put(48,9){$111$}\put(48,-6){$310$}
\put(-11,-25){112}\put(8.5,-25){300}\put(27,-25){010}
\put(71,6){$+$}
}

\put(170,0){
\put(0,-15){\line(0,1){12}}
\put(0,-3){\line(1,0){18.7}}\put(27,-3){\vector(1,0){18}}
\put(15,-15){\line(0,1){27}}\put(15,12){\vector(1,0){30}}
\put(30,-15){\line(0,1){15.5}} \put(30,27){\line(0,-1){18.5}}
\put(30,27){\vector(1,0){15}}
\put(19.5,-5){${\bf u}$}\put(26,2.2){${\bf w}$}
\put(7,15){012}
\put(48,24){$001$}\put(48,9){$111$}\put(48,-6){$301$}
\put(-11,-25){112}\put(8.5,-25){300}\put(27,-25){001}
\put(71,6){$+$}
}

\put(255,0){
\put(0,-15){\line(0,1){12}}
\put(0,-3){\line(1,0){18.7}}\put(27,-3){\vector(1,0){18}}
\put(15,-15){\line(0,1){27}}\put(15,12){\vector(1,0){30}}
\put(30,-15){\line(0,1){15.5}} \put(30,27){\line(0,-1){18.5}}
\put(30,27){\vector(1,0){15}}
\put(19.5,-5){${\bf v}$}\put(26,2.2){${\bf w}$}
\put(7,15){012}
\put(48,24){$001$}\put(48,9){$111$}\put(48,-6){$211$}
\put(-11,-25){112}\put(8.5,-25){201}\put(27,-25){010}
\put(71,6){$+$}
}

\put(340,0){
\put(0,-15){\line(0,1){12}}
\put(0,-3){\line(1,0){18.7}}\put(27,-3){\vector(1,0){18}}
\put(15,-15){\line(0,1){27}}\put(15,12){\vector(1,0){30}}
\put(30,-15){\line(0,1){15.5}} \put(30,27){\line(0,-1){18.5}}
\put(30,27){\vector(1,0){15}}
\put(19.5,-5){${\bf v}$}\put(26,2.2){${\bf w}$}
\put(7,15){012}
\put(48,24){$001$}\put(48,9){$111$}\put(48,-6){$202$}
\put(-11,-25){112}\put(8.5,-25){201}\put(27,-25){001}
}

\end{picture}
\end{equation*}
Here we have set ${\bf u} = 400\in B_4, 
{\bf v}=301 \in B_4$ and ${\bf w} = 001 \in B_1$.
\end{example}

\subsection{Factorization of combinatorial $R$}\label{ss:fcr}

As a preparation for the next subsection, we present a
matrix product formula for the combinatorial $R$.
Let 
$F = \bigoplus_{m \ge 0}\C |m\rangle$
be a Fock space and $F^\ast = \bigoplus_{m \ge 0}\C \langle m |$ be 
its dual with the bilinear pairing such that
$\langle m | m'\rangle = \delta_{m,m'}$\footnote{Ket vectors 
here containing a single integer should not be confused with 
$n$-line states $|{\bf x}\rangle$ nor 
$n$-TAZRP states $|{\boldsymbol \sigma}\rangle$.}.
Let further ${\bf a}^+, {\bf a}^-, {\bf k}$ 
be the linear operators acting on them as
($\langle-1|=|-1\rangle=0$)
\begin{align*}
&{\bf a}^+|m\rangle = |m+1\rangle,\quad
{\bf a}^-|m\rangle = |m-1\rangle,\quad
{\bf k}|m\rangle = \delta_{m,0}|m\rangle,\\
&\langle m| {\bf a}^+=\langle m-1|,\quad 
\langle m|{\bf a}^-=\langle m+1|,\quad 
\langle m|{\bf k}=\delta_{m,0}\langle m|.
\end{align*} 
They satisfy the relations 
\begin{align}\label{lin:uz}
{\bf k}^2 = {\bf k},
\quad{\bf k}\,{\bf a}^+ = 0,
\quad {\bf a}^-{\bf k}=0,
\quad{\bf a}^-{\bf a}^+ = 1,
\quad {\bf a}^+{\bf a}^- = 1-{\bf k},
\end{align}
which coincide with the $q$-oscillator algebra ${\mathcal A}_q$  
at $q=0$ \cite[(2.16)]{KMO}. 
In this sense we refer to ${\bf a}^+, {\bf a}^-, {\bf k}$ as
$q=0$-oscillators and (\ref{lin:uz}) as  
$q=0$-oscillator algebra ${\mathcal A}_0$.
The equality 
$(\left\langle m\right| X){\left| m^{\prime} \right\rangle}
=\left\langle m\right| (X{\left| m^{\prime} \right\rangle})$ 
holds for any $X\in{\mathscr A}_0$.
The ${\mathscr A}_0$ has a basis
\begin{align}\label{lin:kj}
1,\quad ({\bf a}^+)^r,\quad 
({\bf a}^-)^r,\quad
({\bf a}^+)^s{\bf k}\,({\bf a}^-)^t\qquad
(r \in \Z_{\ge 1}, s,t \in \Z_{\ge 0}).
\end{align}
Let ${\mathscr A}^{\mathrm{fin}}_0 \subset {\mathscr A}_0$ 
be the vector subspace spanned by (\ref{lin:kj}) except 1.
Then $\mathrm{Tr}(X) := \sum_{m\ge 0}\langle m|X|m\rangle$
is finite for any $X \in {\mathscr A}^{\mathrm{fin}}_0$.

Introduce the operator $\hat{R}^{a,b}_{i,j} \in \mathrm{End}(F)$ 
together with its diagram representation as
\begin{equation}\label{lin:utkr}
\begin{picture}(150,53)(15,-28)
\thinlines

\put(-77,-7){$\hat{R}^{a,b}_{i,j}\;=$}

\rotatebox{20}{
{\linethickness{0.2mm}
\put(3,1){\color{blue}\vector(-1,0){40}}}}

\put(-0.5,0.2){
\put(-15,-18){\vector(0,1){32}}
\put(-30,0){\vector(3,-1){33}}
\put(-36,-1){$i$}
\put(-16.5,16.5){$b$}
\put(-17,-28){$j$}
\put(6,-15){$a$}}

\put(22,-7){$\displaystyle{=\delta^{a+b}_{i+j}\theta(a \ge j)
({\bf a}^+)^j {\bf k}^{\theta(a>j)}({\bf a}^-)^b}
\qquad (a,b,i,j \in \Z_{\ge 0})$,}
\end{picture}
\end{equation}
where $\delta^a_i=\delta_{a,i}$.
The blue arrow carries the Fock space on which the 
$q=0$-oscillators act. 
The other thin arrows carrying $\Z_{\ge 0}$ should not 
be confused with the thick arrows in (\ref{lin:ok}) 
carrying the elements of crystals.

\begin{lemma}\label{le:linbk}
The matrix element 
$R^{abc}_{ijk} := \langle c|\hat{R}^{a,b}_{i,j}|k\rangle$ is expressed as
\begin{equation*}
\begin{picture}(150,53)(-30,-28)
\thinlines

\put(-82,-7){$R^{abc}_{ijk} \;=$}

\rotatebox{20}{
{\linethickness{0.2mm}
\put(3,1){\color{blue}\vector(-1,0){40}}}}

\put(-0.5,0.2){
\put(-15,-18){\vector(0,1){32}}
\put(-30,0){\vector(3,-1){33}}
\put(-36,-1){$i$}
\put(-16.5,16.5){$b$}
\put(-17,-28){$j$}
\put(6,-15){$a$}
\put(5.5,0){$k$}
\put(-43,-15){$c$}}

\put(22,-7){$\displaystyle{= \delta^a_{j+(i-k)_+}\delta^b_{\mathrm{min}(i,k)}
\delta^c_{j+(k-i)_+}}$,}
\end{picture}
\end{equation*}
where the symbol $(x)_+$ was defined in the beginning of Section \ref{ss:ab}.
\end{lemma}
\begin{proof}
Substituting $\theta(a\ge j) = \delta^a_j + \theta(a>j)$ into (\ref{lin:utkr}) we find
\begin{align}
R^{abc}_{ijk} &= \delta^{a+b}_{i+j}\delta^a_j 
\langle c|({\bf a}^+)^j ({\bf a}^-)^b|k\rangle
+\delta^{a+b}_{i+j}\theta(a>j)
\langle c|({\bf a}^+)^j {\bf k}({\bf a}^-)^b|k\rangle \nonumber\\
&=\delta^{a+b}_{i+j}\delta^a_j \delta^c_{k-b+j}\theta(k \ge b)
+\delta^{a+b}_{i+j}\theta(a>j)\delta^b_k\delta^c_j \nonumber\\
&=\delta^a_j \delta^b_i\delta^c_{k-i+j}\theta(k \ge i)
+\delta^a_{i+j-k}\delta^b_k\delta^c_j\theta(k<i).\label{lin:kmm}
\end{align}
\end{proof}
Note that
\begin{equation}\label{lin:sj}
R^{abc}_{ijk}=0\quad \text{unless}\;\; (a+b, b+c)=(i+j,j+k).
\end{equation}

\begin{proposition}[Matrix product form of combinatorial $R$]\label{pr:lin}
Let ${\bf a} = (a_1,\ldots, a_L),
{\bf i} = (i_1,\ldots, i_L) \in B_\ell$ and 
${\bf b} = (b_1,\ldots, b_L),
{\bf j} = (j_1,\ldots, j_L) \in B_m$.
The matrix element (\ref{lin:ok}) 
of the combinatorial $R$ $R_{\ell,m}$ with $\ell > m$ is expressed as
\begin{equation*}
\begin{picture}(220,82)(-132,-40)

%%%%%%%%%%%%% R %%%%%%%%%%%%%%%%%%%

\put(-10,4){
\put(3,1){\color{blue}\vector(-3,-1){73}}
\put(-210,-8){$R^{{\bf a},{\bf b}}_{{\bf i},{\bf j}} \;= $}

\thicklines
\put(-145,-6){
\put(-18,-2.5){${\bf i}$}
\put(15,-2){${\bf a}$}
\put(-2,-20){${\bf j}$}
\put(-3,14){${\bf b}$}
\put(-10,0){\vector(1,0){20}}
\put(0,-10){\vector(0,1){20}}
}

\thinlines
\put(-110,-8){$=$}
\put(-96,-26){$\mathrm{Tr}\Bigl($}
\put(68,18){$\Bigr)$}

\put(-48,-32){\vector(0,1){34}}
\put(-65,-11){\vector(3,-1){35}}
\put(-51,6){$\scriptstyle{b_1}$}
\put(-74,-10){$\scriptstyle{i_1}$}
\put(-30,-29){$\scriptstyle{a_1}$}
\put(-50,-39){$\scriptstyle{j_1}$}

\put(33,11){
\put(-48,-30){\vector(0,1){32}}
\put(-62,-12){\vector(3,-1){31}}}

\put(-38,0){$\scriptstyle{i_2}$}
\put(-18,17){$\scriptstyle{b_2}$}
\put(-17,-25){$\scriptstyle{j_2}$}
\put(4,-16){$\scriptstyle{a_2}$}

\multiput(5.1,1.7)(3,1){7}{\color{blue}.} 
\put(6,2){
\put(21,7){\color{blue}\line(3,1){30}
}

\put(83,27){
\put(-48,-27){\vector(0,1){26}}
\put(-59,-12){\vector(3,-1){25}}}
\put(15,16){$\scriptstyle{i_L}$}
\put(51,3){$\scriptstyle{a_L}$}
\put(33,29){$\scriptstyle{b_L}$}
\put(34,-7){$\scriptstyle{j_L}$}}

\put(85,-8){$= \mathrm{Tr}
\bigl(\hat{R}^{a_1,b_1}_{i_1, j_1}\cdots \hat{R}^{a_L,b_L}_{i_L, j_L}\bigr)$.}

}
\end{picture}
\end{equation*}
\end{proposition}
\begin{proof}
We are to show
\begin{align*}
R^{{\bf a},{\bf b}}_{{\bf i},{\bf j}}= \sum_{c_1,\ldots, c_L}
R^{a_1b_1c_L}_{i_1j_1c_1}R^{a_2b_2c_1}_{i_2j_2c_2}\cdots
R^{a_{L-1}b_{L-1}c_{L-1}}_{i_{L}j_{L}c_L}.
\end{align*}
The left hand side is $1$ or $0$ depending on whether 
$R({\bf i}\otimes {\bf j}) = {\bf b}\otimes {\bf a}$ or not
according to the NY-rule in Section \ref{ss:cr}.
The right hand side is the sum over $c_1,\ldots, c_L \in \Z_{\ge 0}$
which effectively reduces to a {\em single} sum due to (\ref{lin:sj})
and $\sum_r(a_r, b_r) = \sum_r(i_r, j_r)=(l,m)$.
Consider the dot diagrams in the NY-rule for ${\bf i}, {\bf j}, {\bf a}, {\bf b}$
and their $r$-th boxes from the left which contain
$i_r, j_r, a_r, b_r$ dots, respectively.
We are going to identify $c_{r}$ (resp. $c_{r-1}$) with the numbers of $H$-lines 
coming from the right (resp. outgoing to the left) of these boxes.
The identification is certainly consistent locally since (\ref{lin:kmm}) agrees with  
the NY-rule depicted below under 
the abbreviation $(a,b,c,i,j,k)=(a_r,b_r,c_{r-1},i_r,j_r,c_r)$.
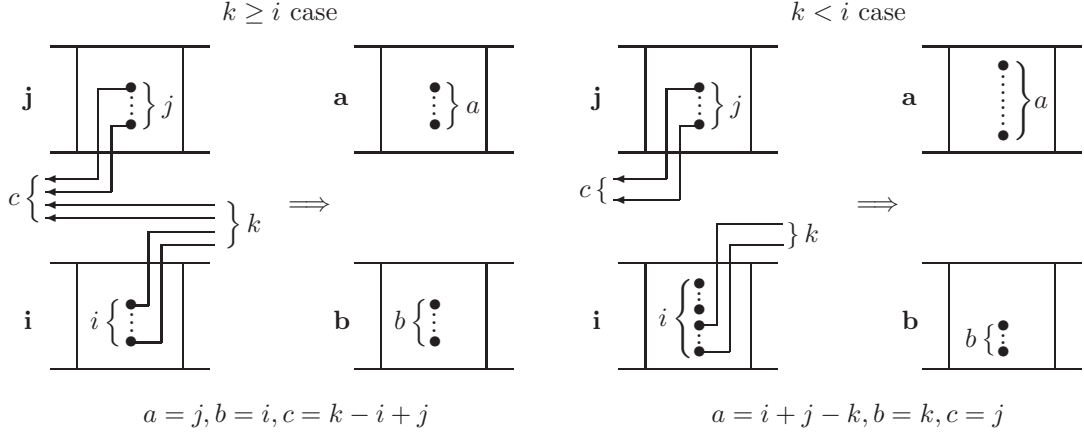
\begin{equation*}
\begin{picture}(400,160)(-13,-20)

%Leftleft
\put(0,82){

\put(65,50){$k\ge i$ case}
\put(0,0){\line(1,0){60}}\put(0,40){\line(1,0){60}}
\put(10,0){\line(0,1){40}}\put(50,0){\line(0,1){40}}
\put(28,8){$\bullet$}
\put(29.3,14.1){.}\put(29.3,17.1){.}\put(29.3,20.1){.}\put(28,22){$\bullet$}
\put(30,24){\line(-1,0){12}}\put(18,24){\line(0,-1){34}}
\put(18,-10){\vector(-1,0){20}}
\put(30,10){\line(-1,0){7}}\put(23,10){\line(0,-1){25}}
\put(23,-15){\vector(-1,0){25}}
\put(13,15){$\left.\phantom{A^{A^A}}\right\} j$}
}

\put(-10,100){${\bf j}$}
\put(-10,15){${\bf i}$}
\put(-16,62){$c\left\{\phantom{A^{A^A}}\right.$}
\put(45,52){$\left.\phantom{A^{A^A}}\right\} k$}
\put(62,62){\vector(-1,0){64}}
\put(62,57){\vector(-1,0){64}}
\put(30,24){\line(1,0){7}}\put(37,24){\line(0,1){28}}\put(37,52){\line(1,0){25}}
\put(30,10){\line(1,0){12}}\put(42,10){\line(0,1){37}}\put(42,47){\line(1,0){20}}
\put(0,0){\line(1,0){60}}\put(0,40){\line(1,0){60}}
\put(10,0){\line(0,1){40}}\put(50,0){\line(0,1){40}}
\put(28,8){$\bullet$}
\put(29.3,14.1){.}\put(29.3,17.1){.}\put(29.3,20.1){.}\put(28,22){$\bullet$}
\put(15,15){$i\left\{\phantom{A^{A^A}}\right.$}

\put(35,-20){$a=j, b=i, c=k-i+j$}
%\Leftright
\put(115,0){
\put(0,82){
\put(0,0){\line(1,0){60}}\put(0,40){\line(1,0){60}}
\put(10,0){\line(0,1){40}}\put(50,0){\line(0,1){40}}
\put(28,8){$\bullet$}
\put(-0.5,0){\put(29.3,14.1){.}\put(29.3,17.1){.}\put(29.3,20.1){.}}
\put(28,22){$\bullet$}
\put(13,15){$\left.\phantom{A^{A^A}}\right\} a$}
}
\put(-25,60){$\Longrightarrow$}
\put(-8,100){${\bf a}$}
\put(-8,15){${\bf b}$}
\put(0,0){\line(1,0){60}}\put(0,40){\line(1,0){60}}
\put(10,0){\line(0,1){40}}\put(50,0){\line(0,1){40}}
\put(28,8){$\bullet$}
\put(-0.5,0){\put(29.3,14.1){.}\put(29.3,17.1){.}\put(29.3,20.1){.}}
\put(28,22){$\bullet$}
\put(15,15){$b\left\{\phantom{A^{A^A}}\right.$}
}

%Rightleft
\put(215,0){
\put(0,82){
\put(65,50){$k<i$ case}
\put(0,0){\line(1,0){60}}\put(0,40){\line(1,0){60}}
\put(10,0){\line(0,1){40}}\put(50,0){\line(0,1){40}}
\put(28,8){$\bullet$}
\put(29.3,14.1){.}\put(29.3,17.1){.}\put(29.3,20.1){.}\put(28,22){$\bullet$}
\put(30,24){\line(-1,0){12}}\put(18,24){\line(0,-1){34}}
\put(18,-10){\vector(-1,0){20}}
\put(30,10){\line(-1,0){7}}\put(23,10){\line(0,-1){28}}
\put(23,-18){\vector(-1,0){25}}
\put(13,15){$\left.\phantom{A^{A^A}}\right\} j$}
}

\put(-10,100){${\bf j}$}
\put(-10,15){${\bf i}$}
\put(-15,65){$c\left\{\phantom{A}\right.$}
\put(57,48){$\left.\phantom{a}\right\} k$}

\put(30,16.5){\line(1,0){7}}\put(37,16.5){\line(0,1){38.5}}
\put(37,55){\line(1,0){25}}
\put(30,6.5){\line(1,0){12}}\put(42,6.5){\line(0,1){40.5}}
\put(42,47){\line(1,0){20}}
\put(0,0){\line(1,0){60}}\put(0,40){\line(1,0){60}}
\put(10,0){\line(0,1){40}}\put(50,0){\line(0,1){40}}
\put(28,20){$\bullet$}\put(29,25.8){.}\put(29,28.2){.}\put(28,30){$\bullet$}
\put(28,4){$\bullet$}\put(29,9.5){.}\put(29,12){.}\put(28,14){$\bullet$}
\put(15,16.5){$i\left\{\phantom{A^{A^{A^{A^{A^A}}}}}\right.$}

\put(35,-20){$a=i+j-k, b=k, c=j$}
%\Rightright
\put(115,0){
\put(0,82){
\put(0,0){\line(1,0){60}}\put(0,40){\line(1,0){60}}
\put(10,0){\line(0,1){40}}\put(50,0){\line(0,1){40}}
\put(0,-1){
\put(28,5){$\bullet$}
\put(29,11.2){.}\put(29,14.2){.}\put(29,17.2){.}\put(29,20.2){.}
\put(29,23.2){.}\put(29,26.2){.}\put(29,29.2){.}\put(28,31.3){$\bullet$}}
\put(9,17){$\left.\phantom{i^{o^{o^{o^{o^2}}}}}\right\} \!a$}
}
\put(-25,60){$\Longrightarrow$}
\put(-8,100){${\bf a}$}
\put(-8,15){${\bf b}$}
\put(0,0){\line(1,0){60}}\put(0,40){\line(1,0){60}}
\put(10,0){\line(0,1){40}}\put(50,0){\line(0,1){40}}
\put(28,4){$\bullet$}\put(29,9.5){.}\put(29,12){.}\put(28,14){$\bullet$}
\put(16,8.5){$b\left\{\phantom{A^{A}}\right.$}
}}
\end{picture}
\end{equation*}

It remains to show that for any given 
${\bf i}\otimes {\bf j} \in B_{\ell}\otimes B_m$ with $\ell>m$, 
there is a unique solution
$(c_1,\ldots, c_L)$ to the simultaneous equations
$c_{r-1}= j_r + (c_r-i_r)_+$ with $r \in [1,L]$ and $c_0 = c_L$.
They are postulated from the rightmost factor $\delta^c_{j+(k-i)_+}$
in Lemma \ref{le:linbk}
and the cyclicity of the trace.
From the relations $c_{r-1}= j_r + (c_r-i_r)_+$ with $r \in [1,L]$, 
we have a piecewise linear expression $c_0=f(c_L)$ in terms of 
$c_L$ including ${\bf i}$ and ${\bf j}$ as parameters.
We are to verify that $x=f(x)$ has a unique solution.
In fact it is given by $x=w:=f(0)$.
To see this note that $f(x+1)=f(x)$ or $f(x)+1$ because 
of $(x+1)_+ = (x)_+$ or $(x)_++1$.
Let $s$ be the smallest nonnegative integer such that
$f(s)= w$ and $f(s+1)=w+1$.
This can happen only if 
$c_{r-1}=c_r+j_r-i_r$ holds for all $r \in [1,L]$ upon the choice $c_L=s$.
Then $f(s)=s+\sum_{r=1}^L(j_r-i_r) = s+m-\ell$.
Thus $f(s)=w$ forces $w<s$ by the assumption $\ell>m$.
Now the unique existence of the solution to $x=f(x)$ is obvious 
from the following graph.

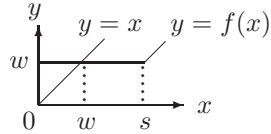
\begin{picture}(70,55)(-140,-11)
\put(0,0){\vector(0,1){30}}
\put(-4,34){$y$}
\put(50,28){$y=f(x)$}
\put(-11,14){$w$}
\put(0,16){\line(1,0){40}}
\multiput(38,1.5)(0,3){5}{.}
\drawline(40,16)(48,24)

\multiput(16,1.5)(0,3){5}{.}
\put(14.5,-9){$w$}

\put(15,28){$y=x$}
\put(38,-9){$s$}
\put(0,0){\line(1,1){25}}
\put(0,0){\vector(1,0){55}}\put(60,-3){$x$}
\put(-6,-9){$0$}
\end{picture}
\end{proof}

From the proof it also follows that 
$R^{a_1,b_1}_{i_1, j_1}\cdots R^{a_L,b_L}_{i_L, j_L} \in {\mathcal A}^\text{fin}_0$
and its trace is convergent.
In fact this fact can directly be derived from (\ref{lin:utkr}) since   
there is at least one $r$ such that $\theta(a_r>j_r)=1$
due to ${\bf a}\in B_\ell, {\bf j}\in B_m$ and  $\ell-m>0$.

Proposition \ref{pr:lin} is a special case $\forall \epsilon_r=0$ of
\cite[Th.6]{Ku} and is also a corollary of \cite[Th.4.1]{KOS} at $q=0$.
The operator $R \in \mathrm{End}(F^{\otimes 3})$
defined by $R(|i\rangle \otimes |j\rangle \otimes |k\rangle)
=\sum_{a,b,c} R^{abc}_{ijk}|a\rangle \otimes |b\rangle \otimes |c\rangle$
using $R^{abc}_{ijk}$ in Lemma \ref{le:linbk} is known to 
satisfy the tetrahedron equation
$R_{1,2,4}R_{1,3,5}R_{2,3,6}R_{4,5,6}=
R_{4,5,6}R_{2,3,6}R_{1,3,5}R_{1,2,4}$ \cite{Zam80}.
Furthermore this $R$ is the $q=0$ limit of 
the 3D $R$ operator including generic $q$,
which has a long history going back to \cite{KV}.
See for example \cite{BS, KOS, Ku} and references therein.
We will present a new application of 
the 3D $R$ operator and the tetrahedron equation to 
$n$-TAZRP in \cite{KMO4}.

\begin{example}
We calculate two elements of the combinatorial $R$: 
$B_4\otimes B_3 \rightarrow B_3 \otimes B_4$ according to 
Proposition \ref{pr:lin}.
\begin{align*}
&R^{1201,0021}_{0121,1101}=\mathrm{Tr}\bigl(
R^{1,0}_{0,1}R^{2,0}_{1,1}R^{0,2}_{2,0}R^{1,1}_{1,1}\bigr) = 
\mathrm{Tr}\bigl(
{\bf a}^+ {\bf a}^+{\bf k}\, ({\bf a}^-)^2 {\bf a}^+ {\bf a}^- \bigr)=
\langle 0|({\bf a}^-)^2 {\bf a}^+ {\bf a}^-{\bf a}^+ {\bf a}^+|0\rangle=1,\\
&R^{1111,0111}_{0121,1101}=\mathrm{Tr}\bigl(
R^{1,0}_{0,1}R^{1,1}_{1,1}R^{1,1}_{2,0}R^{1,1}_{1,1}\bigr) = 
\mathrm{Tr}\bigl(
{\bf a}^+ {\bf a}^+{\bf a}^- {\bf k}\, {\bf a}^- {\bf a}^+{\bf a}^- \bigr)
=\langle 0|{\bf a}^- {\bf a}^+{\bf a}^-{\bf a}^+ {\bf a}^+{\bf a}^-
|0\rangle=0.
\end{align*}
The both elements satisfy the weight conservation, but 
the one matching the combinatorial $R$ is the former.
It coincides with the bottom left vertex in the left hand side of (\ref{ask:dk}). 
\end{example}

\subsection{Matrix product formula for steady state probability}\label{ss:mp}
One may regard $\hat{R}^{a,b}_{i,j}$ (\ref{lin:utkr}) 
as the ${\mathcal A}_0$-valued Boltzmann weight
of a 2D vertex model.
Depict it omitting the blue arrow as 
\begin{equation} \label{mrn:kj}
\begin{picture}(200,45)(-40,-22)
\thinlines
\put(-56,-4){$\hat{R}^{a,b}_{i,j}=$}
\put(-10,0){\vector(1,0){20}}
\put(0,-10){\vector(0,1){20}}
\put(-17,-3.5){$i$}\put(12.5,-3.5){$a$}\put(-2.4,13){$b$}\put(-2.3,-19){$j$}
\put(25,-4){$=\delta^{a+b}_{i+j}
\theta(a \ge j)({\bf a}^+)^j{\bf k}^{\theta(a>j)}({\bf a}^-)^b$.}
\thinlines
\end{picture}
\end{equation}
This vertex $\hat{R}^{a,b}_{i,j}$ made of 
thin arrows carrying $a,b,i,j \in \Z_{\ge 0}$ should be distinguished from 
$R^{{\bf a},  {\bf b}}_{{\bf i}, {\bf j}}$ in (\ref{lin:ok}) 
which consists of thick arrows carrying 
${\bf a}, {\bf i} \in B_\ell$ and ${\bf b}, {\bf j} \in B_m$. 
The factor $\delta^{a+b}_{i+j}$ in (\ref{mrn:kj}) 
represents an ice type conservation rule.
Although it is a vertex model 
whose local states range over the infinite set $\Z_{\ge 0}$,
the quantities relevant to $n$-TAZRP become finite as we will see below.

Recall that a local state $\sigma_i$ of $n$-TAZRP at site 
$i \in \Z_L$ has the form 
$\sigma_i=(\sigma^1_i,\ldots, \sigma^n_i)$ 
in multiplicity representation as in (\ref{cie}).
With such an array $\sigma=(\sigma^1,\ldots, \sigma^n) \in (\Z_{\ge 0})^n$ 
we associate the operator 
$X_\sigma \in \mathrm{End}(F^{\otimes n(n-1)/2})$ defined by
\begin{equation}\label{mrn:ssi}
\begin{picture}(250,90)(-130,-20)
\thinlines

\put(-160,22){$\displaystyle{X_\sigma = X_{\sigma^1,\ldots, \sigma^n} = \sum}$}

\put(51,58){\small $\sigma^n$}
\put(4,43){\small$\sigma^{n-1}+ \sigma^n$}
\put(-54,-3){\small $\sigma^1+\cdots+ \sigma^n$}

\put(0,-15){\line(0,1){12}}\put(0,-3){\vector(1,0){75}}
\put(15,-15){\line(0,1){27}}\put(15,12){\vector(1,0){60}}
\put(45,-15){\line(0,1){57}}\put(45,42){\vector(1,0){30}}
\put(60,-15){\line(0,1){72}}\put(60,57){\vector(1,0){15}}

\put(68,22){$\vdots$}
\multiput(11,20)(2.8,2.8){5}{.}
\put(27,-13){$\cdots$}
\end{picture}
\end{equation}
This is a configuration sum of the ${\mathcal A}_0$-valued vertex model
defined by (\ref{mrn:kj}).
Each edge ranges over $\Z_{\ge 0}$ with the fixed  
boundary condition on the diagonal and the free boundary condition 
on the bottom row and the rightmost column.
The summand represents a {\em tensor product} of the $q=0$-oscillator 
operators (\ref{mrn:kj})  attached to the vertices.
The diagram has the same structure as that in Proposition \ref{th:mnm}.
Note however that the thick arrows there carry elements of crystals whereas
the thin arrows here do just nonnegative integers.
In short the $X_\sigma$ is a corner transfer matrix
of the  ${\mathcal A}_0$-valued vertex model\footnote{
Actually the sum of elements of the corner transfer matrix 
in the original sense \cite{Bax}
since we employ the free boundary condition
on the bottom row and the rightmost column.}.

\begin{example}\label{ex:mrntkk}
For $n=2$ the operator $X_{\sigma}$ with $\sigma=(\sigma^1, \sigma^2)$ 
is given by
\begin{equation*}
\begin{picture}(600,45)(-25,12)
\thinlines
\setlength\unitlength{0.26mm}
\put(50,0){
\put(-7,37){$X_{\sigma^1,\sigma^2}= \ {\displaystyle \sum_j}$}
\put(110,0){
\put(30,25){\line(0,1){35}}
\put(30,60){\vector(1,0){20}}
\put(10,40){\vector(1,0){40}}
\put(10,25){\line(0,1){15}}
\put(18,60){\small{$\sigma^2$}}
\put(-29,40){\small {$\sigma^1+\sigma^2$}}
\put(27,14){\small $j$}
\put(55,37){\small{$j+\sigma^1$}}

\put(100,37){${\displaystyle
= \ \sum_{j\ge0} \,({\bf a}^+)^j{\bf k}^{\sigma^1}({\bf a}^-)^{\sigma^2}}$.}
}
}
\thinlines
\end{picture}
\end{equation*}
For $n=3$ the operator $X_{\sigma}$ with 
$\sigma=(\sigma^1, \sigma^2, \sigma^3)$ is given by
\begin{equation*}
\begin{picture}(600,80)(-59,-35)
\thinlines
\setlength\unitlength{0.26mm}

\put(-3,22){$X_{\sigma^1,\sigma^2,\sigma^3}= \ {\displaystyle \sum_{i,j,k}}$}
\put(170,-30){
\reflectbox{
\rotatebox[origin=c]{90}{
\put(0,40){\line(1,0){50}} \put(50,40){\vector(0,1){23}}
\put(0,20){\line(1,0){30}} \put(30,20){\vector(0,1){43}}
\put(0,0){\line(1,0){10}}
\put(10,0){\vector(0,1){63}}
}}
\put(-41,80){\small{$\sigma^3$}}
\put(-88,60){\small{$\sigma^2+\sigma^3$}}
\put(-136,40){\small{$\sigma^1+\sigma^2+\sigma^3$}}
\put(-30,20){\small$j$}
\put(-50,20){\small$i$}
\put(-25,48){\small$k$}
}

\put(54,-35){$= \ {\displaystyle \sum_{i,j,k}}
({\bf a}^+)^j{\bf k}^{\sigma^1+i-k}({\bf a}^-)^k{\otimes}
({\bf a}^+)^k{\bf k}^{\sigma^2}({\bf a}^-)^{\sigma^3}{\otimes}
({\bf a}^+)^i{\bf k}^{\sigma^1}({\bf a}^-)^{\sigma^2+\sigma^3}$,}
\end{picture}
\end{equation*}
where the sum extends over $i,j\in \Z_{\ge 0}$ and $k \in [0,\sigma^1+i]$. 
Components of the tensor product 
corresponding to the three vertices have been ordered 
as $(\text{bottom right})\otimes (\text{top right}) \otimes 
(\text{bottom left})$.
\end{example}

As seen in these examples, $X_\sigma$ is an infinite sum in general.
However the following formula, which is our second main result in this article, 
is divergence-free.

\begin{theorem}[Matrix product formula 
for steady state probability of $n$-TAZRP]\label{th:aoy}
The steady state probability of the configuration $(\sigma_1,\ldots, \sigma_L)$
of $n$-$\mathrm{TAZRP}$ on the periodic chain $\Z_L$ is expressed as 
\begin{align*}
{\mathbb P}(\sigma_1,\ldots, \sigma_L) = 
\mathrm{Tr}\bigl(X_{\sigma_1}\cdots X_{\sigma_L}\bigr),
\end{align*}
where the trace is taken over $F^{\otimes n(n-1)/2}$.
\end{theorem}

\begin{proof}
Substitute Proposition \ref{pr:lin} into Theorem \ref{th:mnm}
and use the definitions (\ref{sra}) and (\ref{mrn:ssi}).
\end{proof}
Convergence of the trace is guaranteed by the equivalence  
to Theorem \ref{th:mnm} which is manifestly finite.

\begin{example}
Consider $2$-TAZRP on $\Z_4$ in the sector $S(1,1)$.
Translating the configurations e.g. $(\emptyset,\emptyset,1,2)$ 
in multiset representation 
into multiplicity representation $(00,00,10,01)$,  we have
\begin{align*}
&{\mathbb P}(\emptyset,\emptyset,\emptyset,12)=
{\mathbb P}(00,00,00,11)
= \mathrm{Tr}(X_{00}X_{00}X_{00}X_{11})=
\sum\mathrm{Tr}\bigl(
({\bf a}^+)^{j_1+j_2+j_3+j_4}{\bf k}\,{\bf a}^-\bigr)=4,\\
&{\mathbb P}(\emptyset,\emptyset,1,2)=
{\mathbb P}(00,00,10,01)
= \mathrm{Tr}(X_{00}X_{00}X_{10}X_{01})=
\sum\mathrm{Tr}\bigl(
({\bf a}^+)^{j_1+j_2+j_3}{\bf k}({\bf a}^+)^{j_4}{\bf a}^-\bigr)=3,\\
&{\mathbb P}(\emptyset,1,\emptyset,2)=
{\mathbb P}(00,10,00,01)
= \mathrm{Tr}(X_{00}X_{10}X_{00}X_{01})=
\sum\mathrm{Tr}\bigl(
({\bf a}^+)^{j_1+j_2}{\bf k}({\bf a}^+)^{j_3+j_4}{\bf a}^-\bigr)=2,\\
&{\mathbb P}(\emptyset,\emptyset,2,1)=
{\mathbb P}(00,00,01,10)
= \mathrm{Tr}(X_{00}X_{00}X_{01}X_{10})=
\sum\mathrm{Tr}\bigl(
({\bf a}^+)^{j_1+j_2+j_3}{\bf a}^-({\bf a}^+)^{j_4}{\bf k}\bigr)=1,
\end{align*}
where the operators $X_{\sigma^1, \sigma^2}$ are 
taken from Example \ref{ex:mrntkk}. 
The sum is over $j_1,\ldots, j_4 \in \Z_{\ge 0}$.
They agree with $|\xi_4(1,1)\rangle$ in Example \ref{ex:LL}.
\end{example}

\begin{example}
Similarly for $3$-TAZRP on $\Z_3$ in the sector $S(1,2,1)$ we have
\begin{align*}
&{\mathbb P}(1,2,23)=
{\mathbb P}(100,010,011)= \mathrm{Tr}(X_{100}X_{010}X_{011})=
\sum \mathrm{Tr}(Y_1) \mathrm{Tr}(Y_2) \mathrm{Tr}(Y_3), \\
&Y_1=({\bf a}^+)^{j_1}{\bf k}^{1+i_1-k_1}({\bf a}^-)^{k_1}
({\bf a}^+)^{j_2}{\bf k}^{i_2-k_2}({\bf a}^-)^{k_2}
({\bf a}^+)^{j_3}{\bf k}^{i_3-k_3}({\bf a}^-)^{k_3},\\
&Y_2=({\bf a}^+)^{k_1}({\bf a}^+)^{k_2}{\bf k}\,
({\bf a}^+)^{k_3}{\bf k}\,{\bf a}^-,\\
&Y_3=({\bf a}^+)^{i_1}{\bf k}\,({\bf a}^+)^{i_2}{\bf a}^-
({\bf a}^+)^{i_3}({\bf a}^-)^2,
\end{align*}
where the operators 
$X_{\sigma^1, \sigma^2, \sigma^3}$ are again taken from Example \ref{ex:mrntkk}. 
The sum is over $i_r, j_r, k_r \,(r=1,2,3)\in \Z_{\ge 0}$ under the 
condition that all the powers of ${\bf k}$ in $Y_1$ are nonnegative.
There are five such choices yielding the nonvanishing summands as
\begin{align*}
\begin{pmatrix}
i_1 \; j_1\;  k_1\\
i_2 \; j_2 \; k_2\\
i_3 \; j_3 \; k_3
\end{pmatrix}=
\begin{pmatrix}
3 \;\, 1 \;\, 1\\
0 \;\, 0 \;\, 0\\
0 \;\, 0 \;\, 0
\end{pmatrix},
\begin{pmatrix}
3 \;\, 0 \;\, 1\\
0 \;\, 1 \;\, 0\\
0 \;\, 0 \;\, 0
\end{pmatrix},
\begin{pmatrix}
3 \;\, 0 \;\, 1\\
0 \;\, 0 \;\, 0\\
0 \;\, 1 \;\, 0
\end{pmatrix},
\begin{pmatrix}
2 \;\, 0 \;\, 1\\
0 \;\, 1 \;\, 0\\
1 \;\, 0 \;\, 0
\end{pmatrix},
\begin{pmatrix}
2 \;\, 0 \;\, 1\\
0 \;\, 0 \;\, 0\\
1 \;\, 1 \;\, 0
\end{pmatrix}.
\end{align*}
Each of them contributes by 1, reproducing the result 
${\mathbb P}(1,2,23)=5$ in agreement with 
$|\xi_3(1,2,1)\rangle$ in Example \ref{ex:LL} and Example \ref{ex:lin}.
\end{example}

\section*{Acknowledgments}
This work is supported by 
Grants-in-Aid for Scientific Research No.~15K04892,
No.~15K13429 and No.~23340007 from JSPS.

\end{document}